\documentclass[12pt]{article}
\usepackage{caption}
\usepackage{subcaption}
\usepackage{float}
\usepackage{graphicx}
\newcommand{\indep}{\rotatebox[origin=c]{90}{$\models$}}
\usepackage{enumerate}
\usepackage[shortlabels]{enumitem}
\usepackage{bbm}

\usepackage[utf8]{inputenc} 

\usepackage{amsmath,amsfonts,amsthm} 
\usepackage{tikz} 
\usepackage{pgfplots}
\pgfplotsset{compat=1.5}
\usepgfplotslibrary{dateplot}
\usepackage{tikz-cd}
\usetikzlibrary{patterns, shapes, calc}
\usetikzlibrary{positioning,arrows} 
\usetikzlibrary{decorations.pathreplacing} 
\usepackage{tikzsymbols} 
\usepackage{blindtext} 

\usepackage{algorithm}
\usepackage{algorithmicx}
\usepackage{algpseudocode}

\usepackage{pdfpages}
\usepackage[toc,page,header]{appendix}
\usepackage{minitoc}

\usepackage{hyperref}

\usepackage{upgreek}

\usepackage{natbib} 
\usepackage{har2nat} 

\usepackage[english]{babel} 
\usepackage[T1]{fontenc}
\usepackage{lmodern,microtype} 
\usepackage{amsmath,amsthm,amssymb,amsfonts} 
\usepackage{dsfont,mathrsfs,ushort} 
\usepackage{titlesec,titling} 
\usepackage[nohead]{geometry} 
\usepackage{setspace} 
\usepackage{enumitem,booktabs} 
\usepackage{pstricks} 
\usepackage{sgamevar} 
\usepackage{hyperref}

\title{\bfseries\Large Paper Title \vspace*{-1ex}}
\author{\large\itshape Author Name \thanks{}}
\date{\footnotesize \number\month\ $\cdot$\ \number\day\ $\cdot$\ \number\year}

\setenumerate{label=\small(\roman*)}
\hypersetup{colorlinks=true,pdfnewwindow=true,pdfstartview=FitH,%
pdftitle="",pdfauthor="",%
urlcolor=black,citecolor=blue!90!red!45!black,linkcolor=red!90!black}

\gamemathtrue
\titleformat{\section}[block]{\centering\large\bfseries}{\thesection.}{0.5em}{}
\titleformat{\subsection}[block]{\flushleft\bfseries}{\thesubsection.}{0.5em}{}
\titleformat{\subsubsection}[runin]{\normalsize\itshape}{\bfseries\thesubsubsection.}{0.5em}{}[.\:]
\renewcommand{\thesubsubsection}{\arabic{section}.\arabic{subsection}.\alph{subsubsection}}
\titlespacing{\section}{0ex}{6ex}{3ex}
\titlespacing{\subsection}{0in}{3ex}{1.5ex}
\titlespacing{\subsubsection}{0mm}{2ex}{0.5em}
\renewcommand{\linespread}[1]{\setstretch{1}}

\usepackage{mdframed}
\mdfdefinestyle{myenvs}{%
  hidealllines=true,%
  nobreak=true, 
  leftmargin=0pt,
  rightmargin=0pt,
  innerleftmargin=0pt,
  innerrightmargin=0pt,
}

\newenvironment{keyword}{\noindent\textsc{Keywords:} }{}

\newmdtheoremenv[style=myenvs]{prop}{Proposition}[section]
\newtheorem{ass}{Assumption}
\theoremstyle{definition}

\newtheorem*{definition}{Definition}
\newmdtheoremenv[style=myenvs]{theorem}{Theorem}[section]
\newmdtheoremenv[style=myenvs]{lemma}{Lemma}[section]
\newmdtheoremenv[style=myenvs]{alg}{Algorithm}[section]
\newtheorem{remark}{Remark}[section]
\newtheorem{example}{Example}[section]
\newmdtheoremenv[style=myenvs]{cor}{Corollary}

\newtheoremstyle{named}{}{}{\itshape}{}{\bfseries}{.}{.5em}{#1 \thmnote{#3}}
\theoremstyle{named}

\renewcommand{\P}{\mathbb{P}}
\newcommand{\E}{\mathbb{E}}
\usepackage{amssymb}

\title{\vspace{-2cm} Generalized Optimal Transport\thanks{We thank Andres Santos, Denis Chetverikov, Bulat Gafarov, Xiaohong Chen, Rosa Matzkin, Jinyong Hahn, Kirill Ponomarev, Grigory Franguridi, Manu Navjeevan and Daniel Ober-Reynolds for the valuable discussions and criticisms.}}
\author{Andrei Voronin\thanks{Department of Economics, UCLA.  Email: \url{avoronin@ucla.edu}.}}
\date{July 29, 2025}
\begin{document}
\renewcommand{\abstractname}{\vspace{-\baselineskip}}
\doparttoc 
\faketableofcontents 
\parttoc
\maketitle
\begin{abstract}
Many causal and structural parameters in economics can be identified and estimated by computing the value of an optimization program over all distributions consistent with the model and the data. Existing tools apply when the data is discrete, or when only disjoint marginals of the distribution are identified, which is restrictive in many applications. We develop a general framework that yields sharp bounds on a linear functional of the unknown true distribution under i) an arbitrary collection of identified joint subdistributions and ii) structural conditions, such as (conditional) independence. We encode the identification restrictions as a continuous collection of moments of characteristic kernels, and use duality and approximation theory to rewrite the infinite-dimensional program over Borel measures as a finite-dimensional program that is simple to compute. Our approach yields a consistent estimator that is $\sqrt{n}$-uniformly valid for the sharp bounds. In the special case of empirical optimal transport with Lipschitz cost, where the minimax rate is $n^{2/d}$, our method yields a uniformly consistent estimator with an asymmetric rate, converging at $\sqrt{n}$ uniformly from one side.\\

\vspace{.5cm}

\end{abstract}

\newpage
\section{Introduction}\label{s1}

Many causal and structural parameters in economics can be identified and estimated by solving an optimization problem over all distributions consistent with the economic model and the empirical evidence. For example, one may wish to find the smallest value of a long-term average treatment effect that is consistent with both experimental and observational datasets (as in \citet{athey2020combining}, \citet{aizer2024lifetime}, and \citet{obradovic2024identification}), or to assess matching efficiency in a marriage market (as in \citet{chiappori2017partner}) by comparing the realized surplus to the maximum achievable surplus given the observed marginal distributions of partner characteristics. In entry games (e.g. \citet{tamer2003incomplete}, \citet{gurussell}), English auctions (\citet{auctions}) and network formation models (e.g. \citet{miyauchi2016structural}, \citet{de2018identifying}), assessing whether a structural parameter is compatible with the observed distribution of equilibria amounts to checking whether the minimum violation of the model over all data-consistent distributions is zero. In each case, the parameter of interest is the value of an optimization problem over a set of distributions constrained by model structure and data.

This paper introduces a general identification and estimation framework for such problems, which we refer to as \textit{Generalized Optimal Transport} (GOT). Let $T \in \mathbb{R}^d$ denote the vector of variables, both observed and unobserved, that is restricted by an an economic model to a known subset $\mathcal{T} \subseteq \mathbb{R}^d$. Suppose $T$ is distributed according to a true unknown probability measure $\P$, and the parameter of interest is of the form $\E_\P[b(T)]$ for some identified function $b$. Using the theory of characteristic kernels (see \citet{steinwart2021strictly}), we show that two common types of identification conditions -- i) identification of certain joint distributions of subvectors of $T$, and ii) (conditional) independence restrictions on components of $T$ -- can both be represented as moment equalities satisfied by the true distribution $\P$. Sharp bounds on $\E_\P[b(T)]$ can then be obtained by solving optimization problems over all distributions supported on $\mathcal{T}$ that satisfy these conditions. 

GOT accommodates empirical settings that go beyond those typically handled by existing tools, such as linear programming (e.g. \citet{honoreadriana2006}, \citet{santos}, \citet{laffers2019}) or classical optimal transport (\citet{ober2023estimating}, \citet{torous2024optimal}). In particular, it allows for continuously distributed variables, overlapping identified marginals (e.g., $(T_1, T_2)$ and $(T_2, T_3)$, but not their joint), and structural assumptions such as (conditional) independence. These features are central to many modern empirical applications, including \citet{athey2020combining}, \citet{aizer2024lifetime}, \citet{obradovic2024identification}, and \citet{luo2024selecting}.

We develop a computationally simple procedure for finding the value of GOT. First, using duality theory for linear programs (LP) in Banach spaces (\citet{shapiro2001duality}) and the exact penalization approach of \citet{voronin2025linear}, we rewrite GOT as an unconstrained convex program over an $\mathcal{L}^2$ space. Projecting onto a subspace of polynomials of order up to $J$, we obtain a finite-dimensional semi-infinite LP that can be easily solved using cutting-plane methods (see \citet{reemtsen1998numerical}). Crucially, we exploit self-adjointness of the projection operator to derive a Jackson-type bound on the approximation bias of the form $C J^{-r}$, where $C, r > 0$ depend only on known parameters and kernel smoothness.

When the identified components of the model are estimated from data, we construct a computationally simple and consistent estimator of the sharp upper bound\footnote{The same conclusions hold for the lower bound, since $\min_P \E_P[b(T)]=-\max_P-\E_P[b(T)]$.} on $\E_\P[b(T)]$. This estimator converges from below at the uniform rate $\sqrt{n}$, i.e., it is $\sqrt{n}$-uniformly valid for the upper bound, and pointwise consistent from above. The lack of $\sqrt{n}$-uniform consistency from above reflects a positive, vanishing bias term, closely related to the ill-posed inverse problem (see \citet{horowitz2011applied}). Fortunately, the bias is conservative—our estimator does not underestimate the upper bound and overestimate the lower bound uniformly.

GOT nests several important special cases. When $T$ is discrete, it reduces to a linear program and recovers classical identification results (see \citet{balke1994counterfactual, balke1997bounds}, \citet{santos}, \citet{torgovitsky2019partial}, \citet{laffers2019}). When $T$ is supported on a product space and only its disjoint marginals are identified, GOT coincides with multimarginal optimal transport (OT), which has seen growing use in economics (e.g., \citet{galichon2011}, \citet{galichon2016optimal}, \citet{carlier2016vector}, \citet{ober2023estimating}, \citet{gunsilius2025primer}). \citet{manole2024sharp} show that empirical two-marginal OT with a Lipschitz cost suffers from the curse of dimensionality, achieving the minimax rate of $n^{2/d}$. Our estimator remains valid for the upper bound, given by the OT value, at rates up to $\sqrt{n}$. Moreover, for any $\eta \in (0,1]$ chosen by the researcher, we provide an estimator that converges to the OT value at $n^{(1- \eta d/(d+2))/2}$ from below and $n^{\eta/(d+2)}$ from above uniformly.

This work builds on the insight, developed in discrete settings by \citet{balke1994counterfactual, balke1997bounds}, \citet{santos}, \citet{torgovitsky2019partial}, and \citet{laffers2019}, that partial identification problems can often be formulated as optimization problems over distributions. While GOT is also related to the literature on moment (in)equalities (\citet{andrews2013inference, ANDREWS2017275}, \citet{armstrong2014weighted, armstrong2016multiscale}, \citet{chernozhukov2019inference, andrews2023}, \citet{chernozhukovneweysantos2023}), in our case a continuum of moment restrictions involve probability measures as their infinite-dimensional parameters, precluding the application of previously developed methods. There is extensive literature on estimation and inference in the special cases of GOT: the empirical LP (\citet{bhattacharya2009inferring}, \citet{santos}, \citet{semenova2023adaptive}, \citet{chorussel2023}, \citet{gafarov2024simple}, \citet{voronin2025linear}) and the empirical optimal transport (\citet{ober2023estimating}, \citet{manole2024sharp}, \citet{sadhu2024stability}, \citet{hundrieser2024unifying}). More broadly, GOT is related to the literature on partial identification (\citet{MP2000, MP2009}, \citet{beresteanu2011sharp}, \citet{santos}) and estimation of partially identified models (\citet{beresteanu2008asymptotic}, \citet{CHT}, \citet{CLR}).

The rest of the paper is organized as follows. Section \ref{section_identif} introduces the setup, discusses key applications, and develops general identification theory. Section \ref{section_estimation} develops the estimation theory. Section \ref{section_bias} analyzes the vanishing conservative bias term and examines the special case of optimal transport. Section \ref{section_simulation} presents simulation evidence.

\section{Setup and identification}\label{section_identif}
Let $T \in \mathbb{R}^{d}$ be the vector of all economic variables in the model, distributed according to an unknown true probability measure $\mathbb{P}$. The object of interest is the value of a linear functional $\mathbb{E}_P[b(T)]$ evaluated at $P = \mathbb{P}$, where $b$ is an identified cost function. We suppose that two types of identification conditions are available. Firstly, the researcher restricts the support of $T$ to a known compact subset $\mathcal{T} \subseteq \mathbb{R}^d$. Rescaling $T$ if needed, we suppose that $\mathcal{T} \subseteq [-1,1]^d$ without loss of generality. Secondly, for a separable measure space $(S, \mathcal{S}, \nu)$ and identified maps $a$ and $c$, there is a collection of moment restrictions $\mathbb{E}_P[a(T)(s)] = c(s), ~ \nu-\text{a.s. in } s \in S$ that are assumed to be satisfied at $P = \mathbb{P}$. 

The value $\mathbb{E}_{\mathbb{P}}[b(T)]$ of the functional of interest may not be point-identified under the imposed restrictions.  Let $\mathcal{P}$ be the collection of all Borel probability measures supported on $\mathcal{T}$, and define $\theta = (a,b,c)$. Moreover, suppose that $\beta(\theta)$ is the sharp upper bound on the value of interest. It follows that
\begin{align*}
    \text{(GOT)}: \quad \beta(\theta) = \sup_{P \in \mathcal{P}}\mathbb{E}_P[b(T)],  \quad \text{s.t.:} ~ \mathbb{E}_P[a(T)(s)] = c(s) ~ \text{for } s \in S, ~  \nu-\text{a.s.}
\end{align*}
We call the problem (GOT) the generalized optimal transport problem. The rest of this section discusses the applications of (GOT). Any restrictions introduced for that purpose do not apply to the rest of the paper. 
\subsection{Examples}
The simplest applications of (GOT) arise in discrete settings. In case the sets $\mathcal{T}$ and $S$ are finite, (GOT) reduces to a finite-dimensional linear program (LP), a tool long used for partial identification of treatment effects (see, e.g., \citet{balke1994counterfactual, balke1997bounds}, \citet{santos}, \citet{torgovitsky2019partial}, \citet{laffers2019},  \citet{voronin2025linear}).
\begin{example}[LP]\label{example_lp}
    \citet{balke1994counterfactual, balke1997bounds} study identification of the effect of a treatment $D$ on an outcome $Y$, using an instrument $Z$ that affects $Y$ only through $D$ and is independent of unobserved heterogeneity. When $(Y,D,Z) \in \{0,1\}^3$, they characterize agents by a response type $R = (R_y, R_d) \in \{0,1,2,3\}^2$, where $R_d$ indexes latent compliance types (e.g. \textit{compliers}, \textit{defiers}), and $R_y$ potential responses of $Y$ to $D$. For instance, $R_y = 1$ for individuals whose outcome switches from $0$ to $1$ when treated, while $R_y = 2$ captures the reverse. The average causal effect (ACE) is then $\mathbb{E}_\P[\mathds{1}\{R_y = 1\} - \mathds{1}\{R_y = 2\}]$. The sharp upper bound on ACE is the value of a linear program over $p = (P(R_d = i, R_y = j))_{i, j}$, subject to: i) $8$ restrictions of the form $m_{ydz}'p = \mathbb{P}[Y=y, D =d| Z=z] $, with known $m_{ydz}$, ii) total probability one, and iii) non-negativity. This setup can be nested in (GOT) by setting $T = (\mathds{1}\{R_k = j\})_{k \in \{y,d\}, j},$ $S = (s_{ydz})_{y,d,z} =[16],$ with counting measure $\nu$, $b(T) = \mathds{1}\{R_y = 1\} - \mathds{1}\{R_y = 2\},$ $c(s_{ydz}) = \P[Y = y, D= d|Z = z],$ and $a(T)(s_{ydz}) = m'_{ydz}(\mathds{1}\{R_d = i\} \mathds{1}\{R_y =j\})_{i,j}.$
\end{example}
As we show below, (GOT) nests multimarginal optimal transport (OT), which has a variety of economic applications (e.g., \citet{galichon2011}, \citet{galichon2016optimal}, \citet{carlier2016vector}, \citet{luo2024selecting}). Unlike in Example \ref{example_lp}, OT settings often involve uncountable $\mathcal{T}$ and $S$. We defer (GOT) characterizations of the remaining examples to the sequel.
\begin{example}[OT]\label{example_ot}
    Let \( X \in \mathbb{R}^{d_x} \) and \( Y \in \mathbb{R}^{d_y} \) denote buyer and seller characteristics in a market over contracts \( Z \in \mathcal{Z} \), a compact metric space. Buyers with characteristics \( x \) derive utility \( u(x,z) \) from the consumption of $z$, and sellers with characteristics \( y \) incur cost \( v(y,z) \) to produce it. Assume \( u \in C(\mathcal{X} \times \mathcal{Z}) \) and \( v \in C(\mathcal{Y} \times \mathcal{Z}) \) are identified. Suppose also that the marginal distributions \( \mu_X, \mu_Y \) of \( X, Y \) with known compact supports \( \mathcal{X}, \mathcal{Y} \) are identified, and the joint distribution \( \mu^* \) of matched pairs \( (X, Y) \) is also identified. A natural question is whether the realized allocation is a \textit{hedonic} equilibrium, or whether the observed matching is \textit{stable}. \citet{chiapporimccannnesheim} show that both are equivalent to \( \mu^* \) being \textit{efficient}, i.e., achieving the maximal surplus: $
        \int \sup_z \big[ u(x,z) - v(y,z) \big] \, \mathrm{d}\mu^*(x,y) = \sup_{\mu \in \mathrm{M}(\mu_X, \mu_Y)} \int \sup_z \big[ u(x,z) - v(y,z) \big] \, \mathrm{d}\mu(x,y),$
    where \( \mathrm{M}(\mu_X, \mu_Y) \) denotes the set of distributions on \( \mathcal{X} \times \mathcal{Y} \) with marginals \( \mu_X, \mu_Y \). The maximal surplus on the right-hand-side can be estimated using (GOT).
\end{example}

While linear programming and optimal transport are both nested within (GOT), our framework is most useful when their generality falls short—namely, when $T$ is continuously distributed and the identification pattern does not reduce to observing disjoint marginals. (GOT) accommodates settings with \emph{overlapping} identified marginals—e.g., the joint distributions of $T_1, T_2$, and $T_2,T_3$ are identified—and additional structure on $\mathbb{P}$, such as (conditional) independence between components of $T$ (see \citet{galichon2011}, \citet{athey2020combining}, \citet{aizer2024lifetime}, \citet{obradovic2024identification}, \citet{luo2024selecting}).
\begin{example}[Mixed matching]\label{example_mixedmatching}
    As shown by \citet{galichon2011}, many econometric models can be written as $Y \in G(U,X,\alpha_0) ~ \text{a.s.},$ where $Y$ and $X$ are the observed outcome and covariates resp., $U$ is unobserved heterogeneity, $\alpha_0$ is an unknown parameter, and $G(\cdot)$ is a known correspondence. Let $T = (Y,X,U)$ be distributed on a known compact set $\mathcal{T}$. Then $\alpha_0$ is consistent with the model iff there exists a distribution $P_{Y, X, U},$ which agrees with the identified marginal $\P_{Y,X} = P_{Y,X},$ and is supported on a set $\{(y,x,u) \in \mathcal{T}: y \in G(u,x,\alpha_0) \}$. This is equivalent to checking whether the value of the problem, which falls within (GOT), $\inf_{P \in \mathcal{P}} \int  b(y,x,u) \text{d}P(y, x, u), \text{ s.t.: } \P_{Y,X} = P_{Y,X},$ is equal to $0$, where $b(T) = \text{dist}(Y, G(U, X, \alpha_0))$. (GOT) also allows for additional distributional restrictions: for example, assuming that for $X = (X_1, X_2),$ we have $X_1 \indep U$. For optimal transport, this becomes infeasible, when $T$ is continuous (\citet{luo2024selecting}).
\end{example}
\begin{example}[Combining datasets]\label{example_atheychetty}
    Building on \citet{athey2020combining}, suppose the sample consists of two parts: (i) a short-term experimental sample where the joint distribution of a short-term outcome $W \in \mathcal{W}$, treatment $D \in \{0,1\}$, and instrument $Z \in \mathcal{Z}$ is observed, and (ii) a long-term observational sample identifying the joint distribution of the long-term outcome $Y \in [\underline{Y}, \overline{Y}]$ and $D, W$, but not $Z$. Assume unobserved potential outcomes $(V(d))_{d=0,1}$, with $V = DV(1) + (1-D)V(0)$ for $V = W, Y,$ and instrument exogeneity, $Z \indep (Y(j), W(j) )_{j=0,1}$. Potential target parameters include the long-term ATE, $\E[Y(1) - Y(0)]$, the ATE on the treated, $\E[Y(1) - Y(0) \mid D=1]$, and the proportion benefiting from treatment, $\P[Y(1) \geq Y(0)]$. Sharp bounds on each can be obtained using (GOT).
\end{example}
\subsection{Cost function and support restrictions}
Support restrictions in the form of a known compact\footnote{Compactness of $\mathcal{T}$ is not easily relaxed: it underlies the identification of the dual of the space of continuous functions with the space of finite Borel measures, $C^*(\mathcal{T}) \simeq \mathcal{M}(\mathcal{T}),$ on which we rely heavily in Section 3.} support $\mathcal{T}$ may reflect discreteness or boundedness of some components of $T$, and incorporate structural relations. The choice of cost function $b(T)$ determines the target functional of interest $\mathbb{E}_\P[b(T)]$.

\addtocounter{example}{-1}
\begin{example}[cont'd]
    Collect the observed and unobserved economic variables in $T = ((V, (V(j))_{j=0,1})_{V = Y,W}, Z, D)$. $\mathcal{T}$ needs to incorporate both the ex-ante support restrictions on the variables, and structural relationships between them. Thus, we set \begin{align*}
        \mathcal{T} = \{((v, (v(j))_{j=0,1})_{v = y,w}, z, d) \in  [\underline{Y}, \overline{Y}]^3 \times \mathcal{W}^3  \times \mathcal{Z} \times \{0,1\}:& \\
     v = d v(1) + (1-d) v(0), v = y,w \},&
    \end{align*} which is a compact set if $\mathcal{W}, \mathcal{Z}$ are compact. To target the long-run ATE, one sets $b(T) = Y(1) - Y(0)$, while the long-run ATT obtains under $b(T) = \P[D= 1]^{-1}D(Y(1) - Y(0))$. The latter function is usually unknown, but identified. 
\end{example}
    

The estimation procedure for (GOT) will be developed under the assumption of a continuous $b$ -- however, discontinuous functions can often be accommodated by augmenting the vector $T$ and redefining the support accordingly. 
\addtocounter{example}{-1} 
\begin{example}[cont'd]
    To evaluate the proportion of agents who benefit from treatment, one may introduce an auxiliary variable $U = \mathds{1}\{ Y(1) \geq Y(0)\} \in \{0,1\}$, define $T = ((V, (V(j))_{j=0,1})_{V = Y,W}, Z, D, U)$, and let $b(T) = V$, which is continuous. One also augments $\mathcal{T}$ with an additional margin $\{0,1\},$ and a restriction $U = \mathds{1}\{Y(1) \geq Y(0)\}$ \text{a.s.}  
\end{example}
\subsection{Identified marginals}\label{subs_identifmarg}
    Some parts of the distribution of the random vector $T$ may be identified. This information is straightforwardly incorporated into (GOT) in the discrete setting, as shown in Example \ref{example_lp}. We now discuss the more general case. Suppose there is a collection of sets of indices $\mathcal{I} = \{I_i\}^q_{i=1} \subseteq 2^{d}$, such that for any $I \in \mathcal{I}$, the marginal distribution $\mathbb{P}_{I} \equiv \text{proj}_I \mathbb{P}$ is identified. In what follows, we also denote by $T_I$ the subvector of $T$ comprised of variables with indices in $I$. Unlike in the OT setting (e.g. \citet{manole2024sharp} and \citet{sadhu2024stability}), in (GOT) identified marginals need not be disjoint, nor collectively exhaustive, and $q$ may be greater than $2$. 
    
    Let us first consider the case of a single identified marginal distribution $\P_I$. To rewrite this identification restriction in the form of a moment condition in (GOT), we leverage the concept of a characteristic kernel\footnote{We refer the reader to \citet{Wendland_2004}, Chapter 10 for the discussion of kernels and their Reproducing Kernel Hilbert Spaces. Also see \citet{steinwart2021strictly} for a detailed discussion of characteristic kernels.}. 

\begin{definition}
    Suppose $\mathcal{U}$ is a compact metrizable space, $K: \mathcal{U}^2 \to \mathbb{R}$ is a continuous kernel with the Reproducing Kernel Hilbert Space $\mathcal{H}$, and $\mathcal{M}_1(\mathcal{U})$ is the set of all Borel probability measures over $\mathcal{U}$. $K$ is called characteristic if the map $\Phi(P) \equiv \int K(\cdot, \omega)\text{d}P(\omega)$ is injective as an operator $\Phi: \mathcal{M}_1(\mathcal{U}) \to \mathcal{H}$.
\end{definition}


        \begin{remark}
     The function $K(u_1,u_2) = \exp(u_1'u_2)$ is a characteristic kernel. In this case, $\Phi(P)$ gives the moment generating function (MGF) of $U \sim P$. The procedure below may be viewed as an extension of the fact that the distribution of a compactly supported random vector is uniquely determined by the values of its MGF in an open vicinity of $0$. 
    \end{remark}

    
    \begin{remark}
        In the sequel, we also refer to the Matérn kernel, $K: \mathbb{R}^{2q} \to \mathbb{R}$, with $K(x,y) = \frac{2^{1-\eta}}{\Gamma(\eta)}\left(\sqrt{2}\eta||x-y||\right)^{\eta}K_\eta\left(\sqrt{2}\eta ||x-y||\right)$, where $K_\eta(\cdot)$ is the modified Bessel function of second kind and $\eta> 0$ is a parameter. \citet{sriperumbudur2010hilbert} show that it is characteristic.
    \end{remark}
    Intuitively, a characteristic kernel is a device that `separates' probability measures. Let $\lambda$ be the Lebesgue measure over $\mathbb{R}^d$. The following Lemma relies on the fact that continuous functions on a closed cube are equal $\lambda-$a.e. iff they are equal everywhere on it. It is convenient to denote the $d-$cube with side $2$, centered at $t_0$, by $H_d(t_0) \equiv [-1 + t_0 ,1+ t_0]^d$.

    \begin{lemma}\label{lemma_identif}
        Suppose $\mathcal{U} \subseteq \mathbb{R}^{d}$ is a non-empty compact set, and $K: H^2 \to \mathbb{R}$ is a continuous characteristic kernel for a closed cube $H \supseteq \mathcal{U}$ with $\lambda(H)>0$. Then, for $P_1, P_2 \in \mathcal{M}_1(\mathcal{U})$, we have $P_1 = P_2$ iff $\mathbb{E}_{P_1}[K(s,U)] = \mathbb{E}_{P_2}[K(s,U)]$ for $s \in H$ $\lambda-$a.e.  
    \end{lemma}

Fix a continuous characteristic kernel $K: (\text{proj}_I \mathcal{T})^2 \to \mathbb{R}$. Using Lemma 1, observe that 
       \begin{align*}
        P_I = \mathbb{P}_I \iff\mathbb{E}_{P}[K(s_I,T_I)] = \mathbb{E}_{\mathbb{P}_I}[K(s_I,T_I)] ~ \text{for $s \in H_{d}(0)$} ~\lambda-\text{a.e.}
    \end{align*}
    Thus, a restriction $P_I = \mathbb{P}_I$ can be introduced to the problem (GOT) by setting $S \equiv H_d(0)$, with $\mathcal{S}$ being its Borel $\sigma-$algebra, $\nu$ being the Lebesgue measure $\lambda$, and $a(T)(s) \equiv K(s_I, T_I)$ with $c(s) \equiv \mathbb{E}_{\mathbb{P}_I} [K(s_I, T_I)]$ for $s \in S$. 
    
    \addtocounter{example}{-1}
    \begin{example}[cont'd]
       Suppose w.l.g. that $\mathcal{T} \subseteq [-1,1]^d$. To nest the restriction that $\P_{Y,X}$ is identified, let $S = H_{d}(0),$ with $\mathcal{S}$ being its Borel $\sigma-$algebra, and $\nu = \lambda$. Suppose that $I \subseteq 2^{[d]}$ indexes the subvector $(Y,X)$ of $T$, and define $a(T)(s) = K(s_I, (Y,X)),$ and $c(s) = \E_{\P_I}[K(s_I, (Y,X))]$ for $s \in S$, where $K: H_{2|I|}(0) \to \mathbb{R}$ is a continuous characteristic kernel. Observe that such $a$ is known, $c$ is identified, and for any $P \in \mathcal{P},$ we have $P_I = \P_I$ if and only if $\E_P[a(T)(s)] = c(s) ~\nu - \text{a.s.}$ Moreover, if $K(s_I, T_I) = \exp(s_I' T_I),$ and $s_I = (s_x, s_y),$ then $c(s) = \E_{\P_{I}}[\exp( s'_y Y+s'_x X)]$, which is the joint MGF of $Y,X$.
    \end{example}

    We denote by $\iota$ the vector of ones, with dimension inferred from the context. To incorporate many identified marginals simultaneously, one simply constructs $S \equiv \bigcup^{q}_{j = 1} H_{d}(3j)$, with $\mathcal{S}$ being a Borel $\sigma-$algebra over $S$, lets $\nu \equiv \mathcal{\lambda}$, and defines $a(T)(s) \equiv K_j(s_{I_j} - 3\iota j, T_{I_j})$, and $c(s) \equiv \mathbb{E}_{\mathbb{P}_{I_j}}[a(T)(s)]$ for $s \in H_{d}(3j)$ and $T \in \mathcal{T}$, where $K_j: \mathbb{R}^{2|I_j|}\to \mathbb{R}$ is a continuous characteristic kernel, and $j \in [q]$.


    \subsection{Independence restrictions}\label{subs_indeprestr}
In some applications, such as in Examples \ref{example_mixedmatching} and \ref{example_atheychetty}, it may be reasonable to assume independence between subvectors of $T$. The result below allows to incorporate such restrictions into (GOT). Our construction requires identification of the marginal distribution of at least one subvector involved in the independence condition.
    \begin{lemma}
        Suppose $\mathcal{U} \subseteq \mathbb{R}^q$ is a non-empty compact set, $H\supseteq \mathcal{U}$ is a closed ball with $\lambda(H)>0$, and for non-empty $I_1, I_2 \subseteq [q]$, we have $I_1 \cap I_2 = \emptyset$ and $I_1 \cup I_2 = [q]$. Suppose that $K: H^2 \to \mathbb{R}$ is a continuous characteristic kernel. Then, for any $P \in \mathcal{M}_1(\mathcal{U})$, we have $U_{I_{1}} \indep_P U_{I_{2}}$ iff, for  $\text{ for }s \in H, ~ \lambda-\text{a.e.}$,
        \begin{align*}
    \mathbb{E}_P\left[K(s, U) - \int K(s, (U_{I_1},U_{I_2}) )\text{d}P_{I_2}(U_{I_2})\right] = 0. 
        \end{align*}
    \end{lemma}
    
    Suppose that vector $T$ contains disjoint subvectors $T_1$ and $T_2$ with $(T_1, T_2) \in \mathbb{R}^q$, and $s_1, s_2$ are the projections of $s \in \mathbb{R}^d$ onto their respective coordinates. Moreover, the marginal distribution $\mathbb{P}_2$ of $T_2$ is identified. Fixing a continuous characteristic kernel $K: H_{2q}(0)\to \mathbb{R}$, observe that the map $(s,t_1) \to \E_{\P}[K((s_1,s_2),(t_1,T_2))]$ is then also identified. By Lemma 2, to embed the restriction that $T_1 \indep_{\mathbb{P}} T_2$ into (GOT), one may set $S = H_d(0)$, and let $\mathcal{S}$ be the corresponding Borel $\sigma-$algebra, $\nu = \lambda$, and $a(T)(s) \equiv K((s_1,s_2), (T_1, T_2)) - \int K((s_1,s_2),(T_1,T_2))\text{d}\P_2(T_2)$, while $c(s) = 0$ for $s \in S$. Multiple independence restrictions, as well as their combinations with other identifying conditions are accommodated by analogy with the construction in Section \ref{subs_identifmarg}. 
    \setcounter{example}{2}
    \begin{example}[cont'd] Suppose $I_1$ indexes $(Y, D)$, $I_2$ indexes $(W,Z,D)$, $I_3$ indexes $Z$, and $I_4$ stands for $(Y(j), W(j))_{j=0,1}$. To incorporate all identification conditions, let $S = \bigcup^3_{j = 1} H_d(3j)$ with $\mathcal{S}$ being its Borel $\sigma-$algebra, and $\nu = \lambda$. For $j = 0,1$ and $s \in H_d(3j),$ let $a(T)(s) = K(s_{I_j} - 3\iota j,T_{I_j}),$ $c(s) = \E_{\P_{I_j}}[a(T)(s)]$, and for $s \in H_d(0)$, let $a(T)(s +6\iota) = \exp(\sum^4_{j=3}s_{I_{j}}'T_{I_j}) - \exp(s_{I_4}'T_{I_4})\mathbb{E}_{\P_{I_3}}[\exp(s_{I_3}'Z)]$, with $c(s + 6 \iota) = 0$. Thus, the cubes $H_d(0)$ and $H_d(3)$ are used to pin down the identified marginals, while $H_d(6)$ ensures the independence of $(Y(j), W(j))_{j=0,1}$ with $Z$ using the exponential kernel.
    \end{example}
    
    \subsection{Conditional independence restrictions}\label{subs_cond_indep_restr}
    The researcher may be willing to impose conditional independence between subvectors of $T$ - for instance, if a conditionally exogenous instrument is available.  The following Lemma allows to embed conditional independence restrictions into (GOT). As in the unconditional case, the construction requires identification of certain marginal distributions.
    

    \begin{lemma}\label{lemma_cond_indep}
        Suppose $\mathcal{U} \subseteq \mathbb{R}^d$ is a non-empty compact, and for $\{I_j\}^3_{j=1} \subseteq [d]$, we have $I_j \cap I_{k} = \emptyset$ for $j \ne k$, and fix any $\delta > 0$. Then, for any $P \in \mathcal{M}_1(\mathcal{U})$, we have $U_{I_1} \indep_P U_{I_2} | U_{I_3}$ iff, for $s \in \delta H_d(0),~ \lambda - \text{a.e.}$,     \begin{align*}\mathbb{E}_P[\exp(\sum^3_{j=1}s'_{I_j} U_{I_j}) - \mathbb{E}_{P}[\exp(s'_{I_1}U_{I_1})|U_{I_3}] \exp(\sum^3_{j=2}s'_{I_j} U_{I_j})] = 0.
        \end{align*}
    \end{lemma}
    Suppose that vector $T$ contains disjoint subvectors $T_1, T_2, T_3$, and denote by $s_i$ the projections of $s \in \mathbb{R}^d$ onto their respective coordinates. Moreover, suppose that the distribution $\mathbb{P}_{1,3}$ of $(T_1, T_3)$ is identified, so that the conditional MGF of $T_1$ given $T_3$ is also identified. Using Lemma 3, to embed the restriction $T_1 \indep_{\mathbb{P}} T_2 | T_3$ into (GOT), one may set $a(T)(s) = \exp(\sum^3_{j=1}s'_{j} T_{j}) - \mathbb{E}_{\mathbb{P}}[\exp(s'_{1}T_{1})|T_{3}] \exp(\sum^3_{j=2}s'_{j} T_{j})$ and $c(s) = 0$ for an appropriate range of $s$ with Lebesgue measure over it, as discussed in previous sections.

\section{Estimation}\label{section_estimation}
This section develops estimation theory for the problem (GOT). In what follows it will be convenient to reformulate (GOT) as a problem solved over the space of finite positive Borel measures over $\mathcal{T}$, which we denote by $\mathcal{M}^+(\mathcal{T})$. We refer to the following problem as the dual problem throughout the paper:
\begin{align*}
    \text{(D)}: \quad \beta^*(\theta) \equiv \sup_{\lambda \in \mathcal{M}^+(T)} \int b(t) \text{d} \lambda(t), \quad \text{s.t.: ~ ~ } c =  \int a(t) \text{d} \lambda(t)~  \text{ $\nu-$a.s.} 
\end{align*}
Denote the set of solutions of (D) by $\mathcal{A}^*(\theta)$. Further, consider the Hilbert space $X \equiv \mathcal{L}^2(S,\mathcal{S},\nu)$ with the dot product $\langle f, g\rangle \equiv \int f(s)g(s) \text{d}\nu(s)$ for $f, g \in X$. 

\begin{ass}
    i) $a: \mathcal{T} \to X$ and $b: \mathcal{T} \to \mathbb{R}$ are continuous; ii) For some $s^* \in S$ with $\{s^*\} \in \mathcal{S}$ and $\nu(\{s^*\}) > 0$, $a(t)(s^*) = c(s^*) =1$ for all $t \in \mathcal{T}$; iii) (D) is feasible. 
\end{ass}

Assumption 1.i requires the cost function to be continuous, and demands that the constraint map be continuous in $\mathcal{L}^2$ norm. Note that it does not require $a(t)(\cdot)$ to be continuous for all $t \in \mathcal{T}$.  Assumption 1.ii ensures that the problem (D) is equivalent to (GOT), as every feasible measure in (D) should satisfy $\int a(t)(s^*)\text{d}\mu(t) = c(s^*) \iff \int 1\text{d}\mu(t) = 1$. Assumption 1.iii means that the imposed identifying conditions cannot be rejected. 
\subsection{Duality and exact penalization}
Recall that the parameter $\theta$ is not known. In what follows, it will be estimated via some $\hat{\theta}$. The sampling uncertainty in $\hat{\theta}$ renders the direct analysis of (D) complicated, motivating us to approach it using duality for generalized linear programs. Consider the program
\begin{align*}
    \text{(P)}: \quad \beta(\theta) \equiv  \inf_{x \in X} ~ \langle c,x\rangle, \quad \text{s.t.: }  \langle a(t), x \rangle - b(t) \geq 0, ~ \forall t \in \mathcal{T},
\end{align*}
which we refer to as the primary program. 

Under Assumption 1, the problem (D) is the dual to (P)\footnote{Observe that (P) is not necessarily the dual of (D), because the Banach space $\mathcal{C}(\mathcal{T})$ is not reflexive.}. Assumption 1.ii enforces the Slater condition\footnote{See \citet{shapiro2001duality} for the definitions.} in problem (P), ensuring strong duality and solvability of (D) (see e.g. \citet{shapiro2001duality}). One can then adapt the exact penalization idea in \citet{voronin2025linear} to reformulate (P) as an unconstrained optimization problem. Namely, for some $w > 0$, define 
\begin{align}\label{def_unc}
    \beta_\infty(\theta) \equiv \inf_{x \in X} \langle c, x\rangle + w \left(\max_{t \in \mathcal{T}} b(t) - \langle a(t), x \rangle \right)^+.
\end{align}
The objective function in \eqref{def_unc} is called the penalty function, with $\beta_\infty(\theta) \leq \beta(\theta)$ in general. However, if $w$ is larger than the norm of at least one Lagrange multiplier of (P), defined as the solution of (D), then the penalty term $ w \left(\max_{t \in \mathcal{T}} b(t) - \langle a(t), x \rangle \right)^+$ is `large enough' for $x$ that violate the constraints of (P) to ensure that $\beta_\infty(\theta) = \beta(\theta)$. Recall that the solutions of (D) are probability measures under Assumption 1.ii, so their total variation is upper bounded by $1$. Combining these observations leads to the following result. 
\begin{theorem}\label{theorem_dual}
    Let $w = 1$, and suppose that $X, \theta$ satisfy Assumption 1. Then,
    \begin{align*}
        \beta^*(\theta) = \beta(\theta) = \beta_\infty(\theta). 
    \end{align*}
\end{theorem}
We assume that $w \equiv 1$ in what follows.
\begin{remark}
    The fact that the penalty term can be selected at a fixed level, independent of the data generating process, is in sharp contrast with the case of simple LP, where the minimum $w$ that suffices for exact penalization varies with the DGP (see \citet{voronin2025linear}). 
\end{remark}
\subsection{Weak-compactification of $\beta_\infty$}
Unfortunately, even when strong duality holds, $\beta(\theta) = \beta^*(\theta)$ and (D) is solvable, both (P) and the unconstrained problem defined in \eqref{def_unc} may fail to have a solution. To circumvent that, we consider a version of the penalized problem solved over a weak-compact subset of $X$. 
For any $\gamma > 0$, we define $B_\gamma \equiv \{x \in X: ||x||\leq \gamma\}$, and consider
\begin{align*}
    \beta_{\infty}(\theta; \gamma) \equiv \inf_{x \in B_\gamma} ~ \langle c, x\rangle + \left(\max_{t \in \mathcal{T}} b(t) - \langle a(t), x \rangle \right)^+,
\end{align*}
and let the solution set of the above problem be $\mathcal{A}(\theta; \gamma)$. Moreover, suppose the distance between the compactified and the unconstrained versions is $\Delta(\theta; \gamma) \equiv \beta_{\infty}(\theta; \gamma) - \beta(\theta)$. Banach-Alaoglu Theorem and properties of infimum lead to the following Lemma. 
\begin{lemma}\label{lemma_weak_compactif}
    Suppose $X, \theta$ satisfy Assumption 1. Then, i) $\mathcal{A}(\theta; \gamma) \ne \emptyset$ for any $\gamma > 0$, ii) $\Delta(\theta; \cdot)$ is non-increasing, and iii) $\Delta(\theta; \gamma) \downarrow 0$ as $\gamma \to \infty$.
\end{lemma}
\subsection{Passing to polynomials}
While the function $\beta_\infty(\theta;\gamma)$ already allows for perturbation analysis with estimated $\hat{\theta}$, it is computed via infinite-dimensional optimization, and so does not yet constitute a practical estimator. To build the latter, we first formalize the structure of the space $(S, \mathcal{S}, \nu)$. 

\begin{ass}
    For some $m, l\in \mathbb{N}$, for $H_j \equiv H_d(3j), j \in [m]$, and $H_{m+1}$ being a set of size $l$, with $H_{m+1} \cap \bigcup^m_{j =1}\{H_{j}\} =\emptyset$, we have $S = \bigcup_{i = 1}^{m + 1} H_i$. Moreover, $\nu(s) = 1$ for $s \in H_{m + 1}$, and $\nu|_{H_i} = \lambda$ for $i \in [m]$. 
\end{ass}
Assumption 2 allows for continuous restrictions, such as discussed in Section \ref{section_identif}, and accommodates an additional finite collection of moment conditions\footnote{Observe that the measure described in Assumption 2 can be constructed as $\nu(A) = \sum^{m}_{i = 1} \lambda(A \cap H_i) +  l^{-1} \sum_{s \in H_{m + 1}}\mathds{1}\{s \in A\}$ for $A \in \mathcal{S} \equiv \sigma (\bigcup^{m+1}_{i = 1} \mathcal{B}(H_i))$. The exact shape and size of sets $H_i$ is inconsequential, so long as their maximum radius is controlled. For approximation using Assumption 3, it is important that their interiors be Lipschitz domains.}. Whenever Assumptions 1 and 2 are assumed simultaneously, we suppose that $s^* \in H_{m + 1}$.

For a cube $H \subseteq \mathbb{R}^d$, denote by $\mathbf{P}_j(H)$ the space of all polynomials\footnote{While other sieve spaces may also be used, we work with polynomials, as the Jackson-type inequality for them is readily available.} over $H$ of degree up to $j \in \mathbb{N}$. Under Assumption 2, for a collection $J \equiv (k_i)^{m}_{i = 1}$ with $k_i \in \mathbb{N}$ for $i \in [m]$, consider a finite-dimensional closed linear subspace of $X$, given by $\mathbf{P}_{J} \equiv \{f \in X: f|_{H_i} \in \mathbf{P}_{k_i}(H_i), \text{for } i \in [m]\}$. Let $\mathbf{p}_J: X \to \mathbf{P}_J$ be an orthogonal projection onto $\mathbf{P}_J$. It is well-defined by Hilbert projection theorem. Define
\begin{align*}
    \beta_{J}(\theta; \gamma) \equiv \min_{x \in B_\gamma \cap \mathbf{P}_J} \langle c, x \rangle + \left(\max_{t \in \mathcal{T}} b(t) -\langle  a(t), x\rangle \right)^+.
\end{align*}
The value $\beta_J(\theta;\gamma)$ is now computed over a finite-dimensional subspace of $X$. It corresponds to the value of a semi-infinite linear program with $1+l+ \sum^{m}_{i=1} k^d_i$ variables\footnote{Note that this problem is equivalent to $\min_{x^* \geq 0, x \in B_\gamma \cap \mathbf{P}_J} \langle c, x \rangle + x^*~ \text{s.t.:} ~x^* \geq b(t) -\langle  a(t), x\rangle ~ \forall t \in \mathcal{T}$. }, which is feasible computationally even for relatively large $d$, see \citet{reemtsen1998numerical} for details. $\beta_{J}(\theta; \gamma)$ is used to construct our final estimator.

Suppose Assumption 1 holds, and consider $x^* \in B_\gamma \cap \mathcal{A}(\theta, \gamma)$, which is guaranteed to exist by Lemma \ref{lemma_weak_compactif}. Because $\mathbf{p}_J$ is an orthogonal projection,  $||\mathbf{p}_Jx|| \leq ||x||,$ so $\mathbf{p}_{J} x \in B_\gamma \cap \mathbf{P}_J$. Using that and the triangle inequality, one obtains 
\begin{align*}
    \beta_J(\theta;\gamma) - \beta_\infty(\theta;\gamma) \leq |\langle c, (I-\mathbf{p}_J )x^* \rangle| + \text{max}_{t \in \mathcal{T}}|\langle a(t), (I - \mathbf{p}_J)x^* \rangle|.
\end{align*}
Controlling the RHS of the above expression directly is impractical, as $x^*$ may not be smooth, and therefore need not be well-approximated by polynomials\footnote{Of course, because polynomials are dense in $X$, $||(I-\mathbf{p}_J)x^*||_2 \to 0$ as $\min_i k_i \to \infty$ for any fixed $x^*$.}. However, we may use the fact that projection operators are self-adjoint, $\langle x_1, (I - \textbf{p}_J)x_2\rangle = \langle (I - \textbf{p}_J)x_1, x_2\rangle$ for any $x_1, x_2 \in X$. Applying Cauchy-Schwarz inequality then yields
\begin{align}\label{polynoms_approx_1}
    \beta_J(\theta;\gamma) - \beta_\infty(\theta;\gamma) \leq \gamma\left(||(I-\mathbf{p}_J )c||  + \text{max}_{t \in \mathcal{T}}||(I-\mathbf{p}_J )a(t)||\right),
\end{align}
where, if the functions $c$ and $a(t)$ for all $t \in \mathcal{T}$ have appropriate Sobolev smoothness, the right-hand-side is uniformly of polynomial order in $k_i$, see Theorem \ref{theor_general}.

\subsection{Consistency and uniform validity}
For a cube $H$, let $\tilde{H}$ denote its interior, and let $W^{r,2}(\tilde{H})$ be the $r-$order Sobolev space of $\mathcal{L}^2-$functions over $\tilde{H}$, equipped with the Sobolev norm $||\cdot||_{W^{r,2}}$.
\begin{ass}
    Assumption 2 holds, and i) $\exists r \in \mathbb{N}:   ~ c|_{\tilde{H}_i}, a(t)|_{\tilde{H}_i} \in W^{r,2}(\tilde{H}_i) ~ \forall t \in \mathcal{T}$ and $i \in [m]$; ii) $\exists C > 0:$ $\max_{i \in [m]}\max \{||(c|_{\tilde{H}_i})||_{W^{r,2}}, \sup_{t\in \mathcal{T}}||(a(t)|_{\tilde{H}_i})||_{W^{r,2}}\} \leq C$.
\end{ass}
Consider a perturbed version $\hat{\theta} = (\hat{a}, \hat{b}, \hat{c})$ of $\theta$. Under restrctions in Sections \ref{subs_identifmarg}-\ref{subs_cond_indep_restr}, it can be obtained using empirical analogues of $a, c$ and some estimator of $b$. The following Theorem is our main deterministic result, treating $\hat{\theta}$ as fixed. 

\begin{theorem}\label{theor_general}
    Suppose that i) Assumption 2 holds, ii) $X, \theta$ satisfy Assumption 1, iii) $(\hat{a},\hat{b},\hat{c}) \in C(\mathcal{T}, X) \times C(\mathcal{T}) \times X$. Then, 
    \begin{align*}
        -\gamma (\delta^a + \delta^c) - \delta^b \leq \beta_{J}(\hat{\theta}; \gamma) - \beta(\theta) \leq \gamma (\delta^a + \delta^c + \delta^p) + \delta^b + \Delta(\theta; \gamma),
    \end{align*}
    where 
    $\delta^a \equiv \sup_{t \in \mathcal{T}}|| \hat{a}(t) - a(t)||$, $\delta^b \equiv \sup_{t \in \mathcal{T}} |\hat{b}(t) - b(t)|$, $\delta^c \equiv ||\hat{c} - c||$, and $\delta^p \equiv ||(I-\mathbf{p}_{J})c||  + \text{max}_{t \in \mathcal{T}}||(I-\mathbf{p}_{J} )a(t)|| $. If $\theta,X$ also satisfy Assumption 3 and $\kappa \equiv \min_i (J)_i \geq r$, then there exists a constant $\tilde{C} > 0$ independent of $\theta$, such that, $$\delta^p \leq  m C \tilde{C} \kappa^{-r}.$$
\end{theorem}

We now specify the deterministic Theorem \ref{theor_general} to the case of (GOT) under restrictions in Sections \ref{subs_identifmarg}-\ref{subs_indeprestr}. In this case, the parameters $a, c$ are linear functionals of the true unknown probability measure $\mathbb{P}$, and $\hat{\theta}$ is a random map. 
 \begin{ass}\label{ass4}
     For $\P \in \mathcal{P}$ and the set of identified variables' indices $I \subseteq [d]$, the following hold: i)  $a = a(\P_I), c = a(\P_I)$ are constructed, as in Sections \ref{subs_identifmarg}-\ref{subs_indeprestr}, using exponential and/or Matérn kernels with RKHS contained in a Sobolev space of order $r^* \in \mathbb{N}$, and also include an atom $\{s^*\}$ from Assumption 1.ii, while $b = b(\P_I)$; ii) there is an independent sample $\{W_i\}^n_{i=1} \subseteq \mathcal{T}_I,$ with $W_i \sim \P_I$ for $i \in [n]$, and the estimator $\hat{\theta} \equiv \hat{\theta}(\{W_i\}^n_{i=1})$ satisfies iii) $\hat{b}: \mathcal{T}_I^n \to C(\mathcal{T})$ -- Borel-measurable, and iv) $\hat{a} = a(\P_{In}), \hat{c} = c(\P_{In})$.
\end{ass}
Assumption 4 defines $\hat{\theta}$ under the single-sample assumption for notational convenience. It is straightforward to incorporate multiple samples in the spirit of Example \ref{example_atheychetty} by nesting them in a single collection, see \citet{athey2020combining}. When studying uniformity in  $\mathbb{P}$, its support $\mathcal{T}$ is taken to be fixed; the set $I$ and the kernels employed in the construction of $\theta$ in Assumption 4.i are also treated as fixed. 

\begin{theorem}\label{theor_main_cons}
    Suppose that i) Assumption 4 holds for any $\P \in \mathcal{P}$, ii) $\delta^b_n = O_\P(\gamma_n/\sqrt{n})$ uniformly in $\P \in \mathcal{P}$, iii) (GOT) is feasible for any $\P \in \mathcal{P}$. Consider $\gamma_n$ and $\kappa_n \equiv \min_i (J_n)_i ,$ such that $\sup_{\P \in \mathcal{P}} \P[\gamma_n < M \cup \gamma_n > \sqrt{n}M ] = o(1) = \sup_{\P \in \mathcal{P}} \P[n^{-1/2r^*} \kappa_n < M]$ for any $M > 0$. Then, $\Delta(\theta(\mathbb{P}), \gamma_n) = o_{\mathbb{P}}(1)$, and
    \begin{align*}
        O_{\mathbb{P}}\left(\frac{\gamma_n}{\sqrt{n}} \right)\leq \beta_{J_n}(\hat{\theta}; \gamma_n) - \beta(\theta) \leq O_{\mathbb{P}}\left(\frac{\gamma_n}{\sqrt{n}} \right) + o_{\mathbb{P}}(1) ~ \text{for any } \mathbb{P} \in \mathcal{P}.
    \end{align*}
    Moreover, for any $\varepsilon > 0$ and $t_n = o_{\mathbb{P}}\left(\frac{\sqrt{n}}{\gamma_n}\right)$ uniformly over $\mathbb{P} \in \mathcal{P}$, we have
    \begin{align}\label{valid_est}
        \sup_{\mathbb{P} \in \mathcal{P}} \mathbb{P} \left[t_n \left(\beta_{J_n}(\hat{\theta}(\mathbb{P)}; \gamma_n) - \beta(\theta(\mathbb{P}))\right) \leq \varepsilon \right] = o(1).
    \end{align}
\end{theorem}
\begin{remark}\label{remark_4}
    The rate at which $\beta_{J_n}(\hat{\theta};\gamma_n)$ approaches $\beta(\theta(\mathbb{P}))$ from below is uniformly $\sqrt{n}/\gamma_n$. Equation \eqref{valid_est} then means that our estimator is uniformly valid, as it does not underestimate the sharp upper bound on the functional of interest. The rate at which $\beta_{J_n}(\hat{\theta};\gamma_n) \downarrow\beta(\theta)$ depends on the ill-posedness of the population problem, and is not clear in general. If the problem (P) has a solution (e.g. under finite support $\mathcal{T}$), $\Delta(\theta(\mathbb{P}), \gamma) = 0$ for large enough $\gamma$. In Section \ref{section_bias}, we show that $\Delta(\cdot)$ is related to the ill-posed inverse problem that has been studied in the case of nonparametric IV estimation (see \citet{horowitz2011applied}), and derive a uniform rate of $\Delta(\theta;\gamma_n) \downarrow 0$ in the special case of optimal transport. 
\end{remark}
\begin{remark}
    Theorem \ref{theor_main_cons} can be extended to accommodate conditional independence restrictions discussed in Section \ref{subs_cond_indep_restr} by deriving the rates of $\delta^{\bullet}_n$ and using Theorem \ref{theor_general}. This would require additional assumptions that guarantee the continuity of $\E_P[\exp(s_1'T_{1})|T_3]$ in the sense of Assumption 1.i, and nonparametric estimation of the conditional MGFs in Lemma \ref{lemma_cond_indep} will lead to a nonparametric rate of $\delta^a_n \to 0$. Other conditions that satisfy assumptions of Theorem \ref{theor_general} can also be accommodated, which we leave for future research.  
\end{remark}

\begin{remark}
     In the special case of simple optimal transport with Lipschitz cost, the absence of the $\Delta(\cdot)$ term  would imply that the value can be estimated at $\sqrt{n}$ uniformly. This contradicts the minimax rate of $n^{2/d}$, derived by \citet{manole2024sharp}. Our result establishes that the uniform failure of $\sqrt{n}-$convergence in optimal transport is not two-sided. Setting $\gamma_n =O_p(\ln n)$ yields a $\sqrt{n}-$valid estimator (up to log) in the sense of \eqref{valid_est} and Remark \ref{remark_4}, see also Corollary \ref{cor_1}.
\end{remark}

\subsection{Changing fundamentals and many discrete moments}\label{section_highdim}
We now generalize our procedure to some asymptotic regimes, where the set $S = S_n$ and the parameter $\theta = \theta_n$ are allowed to change with $n$. The set $\mathcal{T}$ is treated as fixed.
If Assumption 2 holds, we can decompose any $x = \tilde{x}+ \overline{x} \in X$ into a discrete, $\overline{x} \in \mathbb{R}^l$, and a functional part, $\tilde{x} \in \mathcal{L}^2(\bigcup_{j \in [m]}H_j)$. By analogy, for $a$ and $c$, write $\langle a(\cdot), x\rangle=\langle \tilde{a}(\cdot), \tilde{x} \rangle + \overline{a}(\cdot)'\overline{x}$ and $\langle c, x\rangle = \langle \tilde{c}, \tilde{x}\rangle + \overline{c}'\overline{x}$, where $\langle \cdot\rangle$ also denotes the $\mathcal{L}^2(\bigcup_{j \in [m]}H_j)$ dot-product. Denote $\mathcal{L}_\infty(x;\theta) = \langle c, x \rangle + \left(\max_{t \in \mathcal{T}} b(t) -\langle a(t),x\rangle \right)^+$. When fundamentals $X_n, \theta_n$ change with $n$, the map $\mathcal{L}_\infty$ also changes, which we ignore in notation for simplicity. 

When the number of discrete moments grows with $n$, control over the $\mathcal{L}^2-$norms of the components $\overline{x}, \overline{c}$ may no longer be adequate. For example, one may wish to use Hölder inequality to get $|\overline{x}'(\hat{\overline{c}} - \overline{c})| \leq ||\overline{x}||_z ||\hat{\overline{c}} - \overline{c}||_{z^*},$ where $z^*$ is the conjugate exponent of $z$. If $z^* = + \infty,$ an application of the vector-valued Hoeffding inequality then yields control over the RHS when the dimension of $\overline{c}$ grows. Motivated by this, we consider $\gamma = (\tilde{\gamma}, \overline{\gamma}) \in \mathbb{R}_{++}^2$. Accordingly, for some $z \in \mathbb{N}$, we redefine $B_\gamma \equiv \{x \in X: ||\tilde{x}|| \leq \tilde{\gamma}, ||\overline{x}||_z\leq \overline{\gamma}\}$. If $X$ changes with $n,$ so does $B_\gamma$, which we ignore in our notation. Moreover, $\gamma_n \to \infty$ w.p.a.1 is taken to mean $\min\{\tilde{\gamma}_n, \overline{\gamma}_n\} \to \infty$. The rest of the definitions are used \textit{mutantis mutandis}.

\begin{ass}
    Fix some $z \in [1, + \infty]$. The sequences of true parameters $\theta_n = (a_n, b_n, c_n)$ and measure spaces $(S_n, \mathcal{S}_n, \nu_n)$ generating Hilbert spaces $X_n$ for $n \in \mathbb{N} \cup \{0\}$, satisfy: i) Assumptions 1, 2 for any $n \in \mathbb{N} \cup \{0\}$; ii) $\{\beta(\theta_n)\}^\infty_{n=0}$ is constant; iii) $\forall x_0 \in X_{n-1} ~ \exists x_1 \in X_{n}: ||\tilde{x}_1||_{2}\leq ||\tilde{x}_0||_2$,  $||\overline{x}_1||_z \leq ||\overline{x}_0||_z$ and $\mathcal{L}_\infty(x_1; \theta_n) = \mathcal{L}_\infty(x_0; \theta_{n-1})$.
\end{ass}

Assumption 5 restricts us to asymptotic regimes, in which the feasible sets of the compactified problem corresponding to $\theta_n$ grow monotonically with the sample size in the sense of condition iii), while the optimal value stays fixed.

\begin{lemma}\label{lemma_highdim}
    Suppose $\{X_n,\theta_n\}^\infty_{n=0}$ satisfies Assumption 5. Then, i) $\mathcal{A}(\theta_n; \gamma) \ne \emptyset$ for any $\gamma > 0$, $n \in \mathbb{N} \cup \{0\}$; ii) $\Delta(\theta_n; \gamma)$ is non-increasing in $(\tilde{\gamma}, \overline{\gamma}, n)$ in the sense of product order, and iii) $\sup_{n = 0,1,\dots} \Delta(\theta_n; \gamma) \downarrow 0$ as $\gamma \to \infty$.
\end{lemma}
Suppose that $\hat{\theta}_n = (\hat{a}_n, \hat{b}_n, \hat{c}_n)$ is a deterministic perturbed version of $\theta_n$. The following analogue of Theorem \ref{theor_general} is obtained by combining Hölder inequality and Lemma \ref{lemma_highdim}.
\begin{theorem}\label{theor_general_highdim}
    Suppose i) $\{X_n,\theta_n\}^\infty_{n=0}$ satisfy Assumption 5; ii) ($\hat{a}_n,\hat{b}_n,\hat{c}_n) \in C(\mathcal{T}, X_n) \times C(\mathcal{T}) \times X_n$ for any $n  \in \mathbb{N} \cup {\{0\}}$. Then, for any $n \in \mathbb{N} \cup {\{0\}}$,    
    \begin{align*}
        \sum_{j = a,b,c}\delta^j_n \leq \beta_{J_n}(\hat{\theta}_n; \gamma_n) - \beta(\theta_0) \leq \sum_{j = a,b,c,p}\delta^j_n + \Delta(\theta; \gamma_n),
    \end{align*}
    where 
    $\delta^a_n \equiv \tilde{\gamma}_n\sup_{t \in \mathcal{T}}|| \hat{\tilde{a}}_n(t) - \tilde{a}_n(t)|| + \overline{\gamma}_n\sup_{t \in \mathcal{T}}||\hat{\overline{a}}_n(t) - \overline{a}_n(t)||_{k^*}$, $\delta^b_n \equiv \sup_{t \in \mathcal{T}} |\hat{b}_n(t) - b_n(t)|$, $\delta^c_n \equiv \tilde{\gamma}_n||\hat{\tilde{c}}_n - \tilde{c}_n|| + \overline{\gamma}_n||\hat{\overline{c}}_n - \overline{c}_n||_{z^*}$, $z^* \in [1, + \infty]$ is the conjugate exponent of $z$, and $\delta^p_n \equiv \tilde{\gamma}_n\left(||(I-\mathbf{p}_{J_n})\tilde{c}_n||  + \text{max}_{t \in \mathcal{T}}||(I-\mathbf{p}_{J_n} )\tilde{a}_n(t)|| \right)$. If $\{X_n,\theta_n\}^\infty_{n=0}$ also satisfy Assumption 3 with constants $C, r$, uniformly in $n$, and $\kappa_n \equiv \min_i (J_n)_i\geq r$, then there exists a constant $\tilde{C} > 0$ independent of $\theta_n$, such that, for any $n \in \mathbb{N} \cup \{0\},$
    \begin{align*}
        \delta^p_n \leq  m C \tilde{C} \kappa_n^{-r}.
    \end{align*}
\end{theorem}

Combining Theorem \ref{theor_general_highdim} with the maximal inequality, we derive the convergence rate in the high-dimensional regime, where the number of discrete moments in (GOT), $l_n$, is allowed to grow with $n \to \infty$. In the assumption below, a quantity is `fixed' if it is constant in $n$.

\begin{ass}\label{ass6}
    For $\P \in \mathcal{P}$ and the set of identified variables' indices $I \subseteq [d]$ i) Assumption 2 holds for any $n \in \mathbb{N}$, with $S_n = H_{m+1,n} \cup (\bigcup_{j \in [m]} H_d(j)),$ where $m$ is fixed, and $|H_{m+1,n}| = l_n \geq 1$; ii) There is an independent sample $\{W_i\}^n_{i=1} \subseteq \mathcal{T}_I,$ with $W_i \sim \P_I$ for $i \in [n]$; iii) $\tilde{\theta}(\P_I) = (\tilde{a}(\P_I), \tilde{b}(\P_I), \tilde{c}(\P_I))$ is fixed and satisfies Assumption 4.i, and its estimator $\hat{\tilde{\theta}}$ satisfies Assumptions 4.iii-iv; iv) $\overline{a}_n(T) = \mathbf{B}_n(T_I)$,  $\overline{c}_n(\P_I) = \E_{\P_I}[\overline{a}_{n}(T)]$ for a measurable $\mathbf{B}_n : \mathcal{T}_I \to \mathbb{R}^{l_n}$, such that, for some $\overline{C} > 0$ and all $n \in \mathbb{N}$ and $\P \in \mathcal{P}$,  $||\overline{a}_n(T) - \overline{c}_n(\P_I)||_\infty < \overline{C}$ $\P - $a.s. and $(\mathbf{B}_n(\cdot))_{l_n} = 1$; v)$\hat{\overline{a}}_n = \overline{a}_n$, while $\hat{\overline{c}}_n = \overline{c}_n(\P_{In})$;
\end{ass}
Assumption \ref{ass6} formalizes the construction of $\theta_n$ that combines a fixed collection of conditions in Assumption \ref{ass4} with a growing discrete collection of moments of some known bounded functions of $T$. Boundedness of $\textbf{B}_n$ may likely be relaxed to a restriction on tail behavior, which we leave for future research.
\begin{theorem}\label{theor_main_cons_highdim}
 Let $R_n \equiv (\tilde{\gamma}_n + \overline{\gamma}_n \log l_n)/\sqrt{n}.$ Suppose that i) Assumption 6 holds for any $\P \in \mathcal{P}$, ii) $\delta^b_n = O_\P(R_n)$ uniformly in $\P \in \mathcal{P}$, iii) (GOT) is feasible for any $\P \in \mathcal{P}$ and $n \in \mathbb{N}$, and iv) $\{\beta(\theta_n)(\P)\}_{n\in\mathbb{N}}$ is a constant sequence for any $\P \in \mathcal{P}$. Suppose $z = 1$, while $\gamma_n$ and $\kappa_n \equiv \min_i (J_n)_i ,$, are such that $\sup_{\P \in \mathcal{P}} \P[\tilde{\gamma}_n < M \cup R_n > M ] = o(1) = \sup_{\P \in \mathcal{P}} \P[n^{-1/2r^*} \kappa_n < M]$ for any $M > 0$.  Then, $\Delta(\theta_n(\mathbb{P}), \gamma_n) = o_{\mathbb{P}}(1),$ and
    \begin{align*}
        O_{\mathbb{P}}\left(R_n \right)\leq \beta_{J_n}(\hat{\theta}; \gamma_n) - \beta(\theta_1) \leq O_{\mathbb{P}}\left(R_n\right) + o_{\mathbb{P}}(1) ~ \text{for any }\P \in \mathcal{P}.
    \end{align*}
    Moreover, for any $\varepsilon > 0$ and $t_n = o_{\mathbb{P}}\left(R_n\right)$ uniformly over $\mathbb{P} \in \mathcal{P}$, we have
    \begin{align*}
        \sup_{\mathbb{P} \in \tilde{\mathcal{P}}} \mathbb{P} \left[t_n \left(\beta_{J_n}(\hat{\theta}(\mathbb{P)}; \gamma_n) - \beta(\theta_1(\mathbb{P}))\right) \leq \varepsilon \right] = o(1).
    \end{align*}
\end{theorem}
\begin{remark}\label{remark10}
    Theorem \ref{theor_main_cons_highdim} allows to augment the procedure in Sections \ref{subs_identifmarg}-\ref{subs_indeprestr} with a growing discrete collection of moments of identified bounded functions. For example, if $H_{m+1} = \{s_{j}\}^{(l_n + 1)}_{j=1}$ with $s_{l_n +1} = s^*$ being the atom from Assumption 1.ii, and $T = (T_1, T_2)$, where $T_1$ is observed, one may set $a(\cdot)(s_j) = \mathrm{B}_j(\cdot)$, and $c(s_j) = \E_\P[\mathrm{B}_j(\cdot)]$ for $j \in [l_n]$, where $\mathbf{B}_{n} = (\mathrm{B}_j)^{[l_n]}_{j=1}$ is a spline basis over the support of $T_1$.
\end{remark}
\begin{remark}
    Condition iv) in Theorem \ref{theor_main_cons_highdim} requires the added discrete moments to be \textit{redundant for identification}. In the example in Remark \ref{remark10}, this is enforced by including the moments of a characteristic $K(\cdot)$ over the support of $T_1$ in $\tilde{a},\tilde{c}$. Such moments may be relevant for the $\Delta(\cdot)$ term, which we leverage in Section 4. Condition iv) may be relaxed to $\beta(\theta_n) \downarrow \beta(\theta_1)$ at the cost of an additional vanishing term -- however, its rate is unclear. 
\end{remark}
\section{Characterization of $\Delta(\theta,\gamma)$}\label{section_bias}
Define $\mathbf{A}: \mathcal{L}^2(\nu) \to \mathcal{C}(\mathcal{T})$ to be the bounded linear operator mapping $x \to \langle a(\cdot), x \rangle$, and let $R(\mathbf{A})$ be its range. Under Assumption 1, we can then rewrite the problem $(P)$ as 
\begin{align*}
    \inf_{x \in X} ~ \mathbb{E}_\P [ (\mathbf{A}x)(T)], \quad \text{s.t.: } ~ (\mathbf{A}x)(T) \geq b(T), ~ \forall T \in \mathcal{\mathcal{T}}.
\end{align*}
In the proposition below, we treat $\gamma$, and $B_\gamma$ as in Section \ref{section_highdim}, and $\iota$ is the vector of ones.
\begin{prop}\label{theor_charact}
    Suppose that i) Assumption 1 holds; ii) $\mathcal{T}$ has positive Lebesgue measure; iii) for some $\mathbb{P} \in \mathcal{P}$, we have $c(s) = \mathbb{E}_{\mathbb{P}}[a(T)(s)]$ ~ $\nu -$a.e., and iv) $\mathbb{P} \ll \lambda$, with the Radon-Nikodym derivative $\frac{\partial \P}{\partial \lambda} \in \mathcal{L}^2(\lambda)$. We have
    \begin{align}\label{eq_solvp}
        \beta(\theta) = \inf_{f \in \mathcal{L}^2( \lambda)} ~ \mathbb{E}_\P[f(T)], \quad \text{s.t.: } f \in \text{Cl}\left(R(\mathbf{A}) \cap \{f \in \mathcal{L}^2(\lambda): ~f \geq b~ ~ \lambda-\text{a.e.}\}\right), 
    \end{align}
    where the closure is in $\mathcal{L}^2(\lambda)$. If the infimum is attained by some $f^* \in \mathcal{L}^2(\lambda)$ and $\eta(\gamma) \equiv \inf_{x \in B_\gamma} \operatorname{ess}\sup |\mathbf{A} x - f^*|$, then, for any $\gamma >0$, we have 
    \begin{align*}
        \Delta(\theta, \gamma + \iota \nu(s^*) \eta(\gamma)) \leq \eta(\gamma)(\lambda(\mathcal{T}) + \nu(s^*)).
    \end{align*}
\end{prop}
\begin{remark}
    Proposition \ref{theor_charact} suggests that while the solution to (P) may exist in the closure of the range of the operator $\mathbf{A}$, the inverse $\mathbf{A}^{-1}f^*$ may not exist. This occurs, because the range of a bounded linear operator mapping into an infinite-dimensional space is not necessarily closed. This problem is known as the ill-posed inverse problem in the context of nonparametric IV estimation (see \citet{horowitz2011applied}). 
\end{remark}

\subsection{Special case: Optimal Transport}
Proposition \ref{theor_charact} suggests that to control $\Delta(\theta, \gamma)$ in terms of $\gamma$, one needs to study how well the solution $f^*$ (if it exists) can be approximated by functions in the range of $\textbf{A}$. We now provide further evidence on that issue in the special case of optimal transport. 

\begin{ass}
    i) $T = (T_1,T_2) \in \mathbb{R}^{d}$, where $T_i \in \mathbb{R}^{d_i}$ for $i = 1,2$ are two random subvectors, marginally distributed according to $\mathbb{P}_i$ and supported on compact subsets $\mathcal{T}_i \subseteq H_{d_i}(0)$, with $\mathcal{T} = \mathcal{T}_1 \times \mathcal{T}_2$; ii) for some $L  > 0$, $b$ is $L-$Lipschitz over $\mathcal{T}$.
\end{ass}
Under Assumption 7, let $\mathbf{M}(\mathbb{P}_1, \mathbb{P}_2) \subseteq \mathcal{T}$ be the set of all probability distributions over $\mathcal{T}= \mathcal{T}_1 \times \mathcal{T}_2$ with marginals $\mathbb{P}_1, \mathbb{P}_2$. The classical optimal transport problem is
\begin{align*}
    \text{(OT)} ~ \max_{P \in \mathbf{M}(\mathbb{P}_1,\mathbb{P}_2)} \mathbb{E}_P[b(T)].
\end{align*}
For a Lipschitz-continuous cost function $b(\cdot)$, because the measures $\mathbb{P}_i$ are compactly supported under Assumption 7, a standard duality result asserts that the value of (OT) is equal to the minimum of the following dual Kantorovich problem
\begin{align*}
    \text{(DOT)} \min_{f_i \in C(\mathcal{T}_i), ~ i = 1,2} \mathbb{E}_{\mathbb{P}_1}[f_1(T_1)] + \mathbb{E}_{\mathbb{P}_2}[f_2(T_2)], \quad \\\text{s.t.: } ~ b(T_1, T_2) \leq f_1(T_1) + f_2(T_2), ~ \forall (T_1, T_2) \in \mathcal{T},
\end{align*}
where the minimum is \textit{attained} by some Lipschitz-continuous dual potentials $f^*_{i}$ (see Theorem 5.10 and p.68 in \citet{villani2008optimal} and Lemma E9 in \citet{ober2023estimating}). Smoothness of these potentials can then be used to characterize the error $\Delta(\cdot)$ using arguments similar to those in Proposition \ref{theor_charact}. 

\subsubsection*{Nesting (OT) in (GOT) under Assumption 7} We use the procedure discussed in Section \ref{subs_identifmarg}, augmenting it with a collection of discrete moment conditions, as described in Section \ref{section_highdim}. Write $s = (s_1, s_2) \in \mathbb{R}^d$, with $s_i \in \mathbb{R}^{d_i}$. Let $\mathbf{B}^{(k)}_i: H_{d_i}(0) \to \mathbb{R}^{k^{d_i}}$ be the basis of $k^{d_i}$ equally-spaced normalized linear splines over $H_{d_i}(0)$. Suppose $K_i: \mathbb{R}^{2d_i} \to \mathbb{R}$ is a fixed continuous characteristic kernel, and $a(T)(s) = K_i(s_i - 3i,T_j)$ for $s \in H_d(3i)$, $T \in \mathcal{T}$. Moreover, fix $H = \{s^*\} \cup \{h_j\}^{l}_{j = 1} \subseteq \mathbb{R}^d$, such that i) $|H| = l + 1$, and ii) $\{H, \{H_d(3i)\}_{i=1,2}\}$ are disjoint, and iii) $l = k_1^{d_1} + k_2^{d_2}$ for some $k_1, k_2\in \mathbb{N} \cup \{0\}$. Let $(a(T)(h_j))^l_{j =1} = (\mathbf{B}^{(k_i)}_i(T_i))_{i=1,2}$, while the point $s^*$ is as in Assumption 1. Finally, construct $S$ as in Assumption 2, and set $c(\cdot)$ to be the identified expectation of $a(T)(\cdot)$. Understanding that $k_i = k_{in}$ for $i = 1,2$, let $\theta^{(ot)}_n$ and $\mathbf{A}^{(ot)}_n$ be the resulting parameter and the linear operator.
\begin{lemma}\label{lemma_splines}
    Under Assumption 7, if $k_{in} > 1$ for $i=1,2, n \in \mathbb{N}$, then for any $n \in \mathbb{N},$ and any $\varepsilon > 0$, we have, for any $\{\tilde{\gamma}_n\}_{n\in \mathbb{N}}$,
    \begin{align*}
        \Delta(\theta^{(ot)}_n, \gamma_n) \leq 2L\sum^2_{i=1} \frac{\sqrt{d_i}}{k_{in} - 1},
    \end{align*} where $\overline{\gamma}_n   = (2||b||_\infty + L)(2+\sum^2_{i=1}k_{in}^{d_i})$, and $\ell^1-$norm of $\overline{x}$ is controlled.
\end{lemma}

Lemma \ref{lemma_splines} thus yields a uniform decay rate for $\Delta(\cdot)$ term in the special case of optimal transport, provided we consider asymptotics $k_{i} \to \infty$.The following result is a straightforward combination of Lemma \ref{lemma_splines} and the results in Section \ref{section_highdim}. We develop it for $d_1 = d_2$ and $k_{1n} = k_{2n} \equiv k_n$. We also suppose that $\ell^1$ norm of $\overline{x}$ is controlled, i.e. $z = 1$.
\begin{cor}\label{cor_1}
    Suppose that i) Assumption 7 holds for any $\P \in \tilde{\mathcal{P}} \subseteq \mathcal{P},$ with the uniform bound  $\sup_{\P \in \tilde{\mathcal{P}}} ||\max \{b(\P), L(\P)\}||_\infty < \infty$; ii) $d_1 = d_2$; iii) $\kappa_n /n^{1/2r^*} \to \infty$, where $r^* \in \mathbb{N}$ is a lower bound on the Sobolev smoothness order of the used kernels; iv) for some $\eta \in (0;1],$ $\overline{\gamma}_n \asymp \tilde{\gamma}_n \asymp k^{d/2}_n \asymp n^{\eta d/2 (d+2)}$. If $\hat{\theta}^{(ot)}_n$ is an empirical analogue of $\theta^{(ot)}_n$, then, for any $\epsilon, \varepsilon > 0,$ 
    \begin{align*}
        \sup_{\P \in \tilde{\mathcal{P}}} ~  &\P\big[n^{\frac{1}{2}(1-\frac{\eta d}{d+2}) - \epsilon} (\beta_{J_n}(\hat{\theta}^{(ot)}_n;\gamma_n) - \beta(\theta^{(ot)}_1(\P)) > \varepsilon\big] +  \\
         &\P\big[n^{\frac{\eta}{d+2} - \epsilon} (\beta_{J_n}(\hat{\theta}^{(ot)}_n; \gamma_n) - \beta(\theta^{(ot)}_1(\P)) < - \varepsilon\big] = o(1).
    \end{align*}
\end{cor}
\begin{remark}
     By varying $\eta \in (0,1]$ in Corollary \ref{cor_1}, one can regulate the asymmetry of the rates at which estimator approaches $\beta(\theta_1(\P))$ from two sides. Importantly, if $\beta(\theta_1(\P))$ estimates an upper bound on the parameter, at  $\eta  \approx 0$ one gets an approximately $\sqrt{n}-$valid estimator. Setting $\eta = 1$ yields a symmetric uniform rate of $1/(d+2)$ up to log terms. However, it is known that the minimax two-sided rate of $2/d$ is achievable in this setting, so in case two-sided convergence is prioritized, one should rely on \citet{manole2024sharp} instead. Note that while the LP arising in empirical OT has $O(n^2)$ variables and $O(n)$ restrictions, our problem, albeit semi-infinite, has at most $O(n^{1/2})$ variables. In our simulations, it is computed at the cost of around $50$ LPs with few constraints.
\end{remark}
\begin{remark}
    Lemma \ref{lemma_splines} and the results of \citet{manole2024sharp} may be taken to suggest that using the characterization in Proposition \ref{theor_charact} may lead to a faster two-sided rate than in Theorem \ref{theor_main_cons_highdim}, or to a characterization of $\Delta(\cdot)$ in a more general case. Unfortunately, this is infeasible at this stage for two reasons. Firstly, it is only in the special case of (multimarginal) optimal transport that one can directly characterize the set $\text{Cl}\left(R(\mathbf{A}) \cap \{f \in \mathcal{L}^2(\lambda): ~f \geq b~ ~ \lambda-\text{a.e.}\}\right)$, because one knows that the range of $\mathbf{A}$ is dense in $\{f(t_1,t_2) = g(t_1) + h(t_2): g \in C(\mathcal{T}_1), h \in C(\mathcal{T}_2) \}$. Such characterization becomes infeasible, once the model includes, for example, independence restrictions, as in Sections \ref{subs_indeprestr}-\ref{subs_cond_indep_restr}. Secondly, both previous methods (e.g. \citet{ober2023estimating}), and our characterization of $\Delta(\cdot)$ in Lemma \ref{lemma_splines} rely on the smoothness of the solutions $f^*$ to \eqref{eq_solvp}, which has not been established as a general property of such problems outside of classical OT. For instance, only recently \citet{gladkov2019multistochastic} showed that in a multistochastic OT over three scalar variables, cost $xyz$, and fixed $U[0,1]^2$ pairwise marginals, there exist Lipschitz $f^*$.
\end{remark}

\section{Simulation evidence}\label{section_simulation}
Consider a simple example, in which
$T = (T_1, T_2, T_3) \in \mathcal{T} = [0,1]^3$. Accordingly, $s = (s_1, s_2, s_3) \in \mathbb{R}^3$. Suppose that the marginal distributions $\P_{1,2}$ and $\P_{2,3}$ are identified.  Namely, suppose that $T_i \sim U[0,1]$ for $i = 1,2,3$, and $T_1 \indep T_2$, and $T_2 \indep T_3$. Consider
\begin{align}\label{simple_problem}
    \sup_{P \in \mathcal{P}} \int T_3 - T_1 \text{d} P(T), \text{s.t.: } P_{1,2}  = \P_{1,2}, P_{2,3} = \P_{2,3}.
\end{align}
It is clear that the value of \eqref{simple_problem} is identified and equal to $\E_{\P_{3}} T_3 - \E_{\P_1} T_1 = 0$. For simplicity, we generate one i.i.d. sample $\{T_i\}^n_{i = 1}$ per simulation. Accordingly, we define the oracle estimator to be $\hat{\beta}^{oracle}\equiv \frac{1}{n}\sum^{n}_{i = 1} T_{3i} - T_{1i}$.

\begin{figure}[h!]
    \centering
    \includegraphics[width=0.7\linewidth]{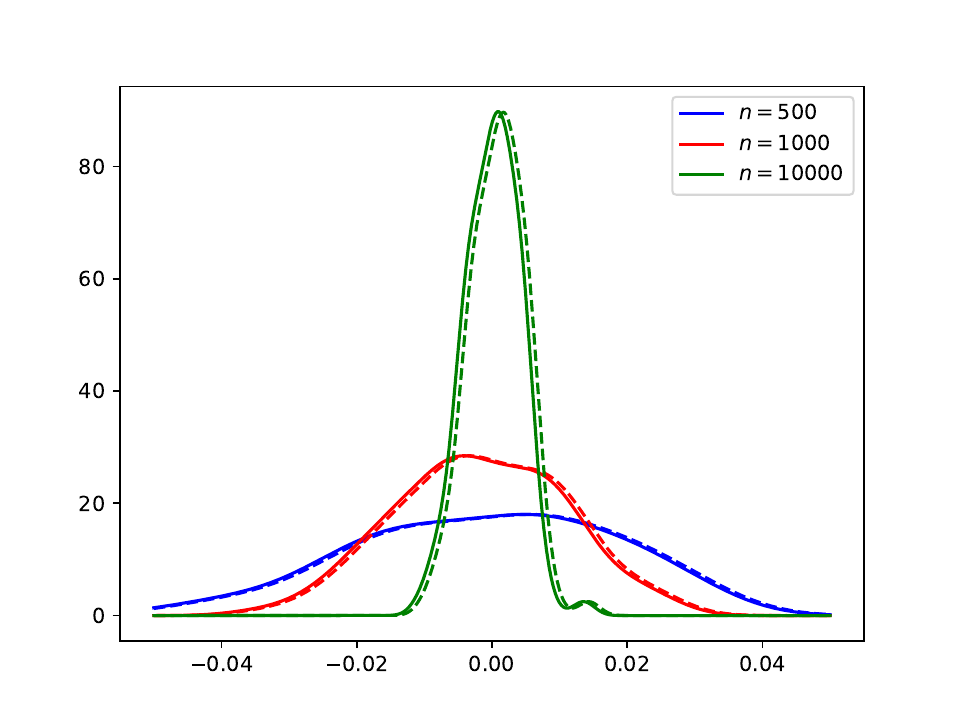}
    \caption{Example \eqref{simple_problem}, KDEs computed over a $100$ simulations for oracle (solid) and GOT (dashed), fixed $\gamma = 5$, and $J = 3$. GOT solved using cutting-plane method with tol = $10^{-6}$.}
    \label{fig:enter-label}
\end{figure}

To apply our estimation procedure, we set $b(T) = T_3 - T_1$, $a(T)(s) = \exp(s_1 T_1 + s_2T_2)$ for $s \in [-1,1]^3$, $a(T)(s) = \exp((s_2 - 3)T_2 + (s_3-3)T_2)$ for $s \in [-1,1]^3 + 3$. $c(s) = M_{1}(s_1,s_2) - $ the MGF for $s \in [-1,1]^3$, and $c(s) = M_2(s_1-3,s_2-3)$ for $s \in [-1,1]^3 + 3$. $\nu = \lambda$, and $s^* = (5,5)$. Denote $\psi_{ijk}(s) \equiv s^i_1s^j_2s_3^k$, and $H \equiv [-1,1]^3$. Performing the polynomial expansion of degree $J \in \mathbb{N}$ for each dimension of $T$, we obtain the problem 
\begin{align}\label{pop_prob}
    \beta_J(\theta; \gamma) \equiv 
    \min_{(x, y, x^*) \in B_\gamma, a \geq 0} \Big \{~ a + x^* +  \\\notag\sum_{(i,j,k) \in \overline{0,J}^3} x_{ijk} \int_{H} M_1(s_1,s_2) \psi_{ijk}(s)  \text{ds} + y_{ijk} \int_{H} M_2(s_2,s_3) \psi_{ijk}(s)  \text{ds} \Big \}, \\\notag
    \text{s.t. for }T \in [0,1]^3, ~ a + x^* \geq T_3 - T_2 -\\\notag
    \sum_{(i,j,k) \in \overline{0,J}^3} \left( x_{ijk} \int_{H} \exp(T_{1,2}'s_{1,2}) \psi_{ijk}(s)  \text{ds} + y_{ijk} \int_{H} \exp(T_{2,3}'s_{2,3}) \psi_{ijk}(s)  \text{ds} \right),
\end{align}
where $B_\gamma$ is obtained by restricting $||(x, y, x^*)||_\infty < \gamma$ for simplicity\footnote{The resulting space is still a convex bounded subset of $X$, and so all results in Section 3 apply.}. The sample analogue of the problem \eqref{pop_prob}, with the value of $\beta_J(\hat{\theta}; \gamma)$, is then obtained by substituting the true MGF. $M_1(s_1,s_2)$ with its plug-in estimator, $\hat{M}_{1}(s_1,s_2) \equiv n^{-1}\sum^n_{i = 1} \exp(T_{1i}s_1 + T_{2i} s_2)$ for all $s_{1}, s_2 \in [-1,1]^3$, and similarly for $M_2(\cdot)$. 

Using the exponential kernel allows to avoid numerical integration. We have 
\begin{align*}
    \int_{[-1,1]^3} \exp(T_{1,2}'s_{1,2}) \psi_{ijk}(s)\text{d} s = \frac{2\mathds{1}\{k \mod 2 = 0\}}{k+1} I_i(T_1) I_{j}(T_2), 
\end{align*}
where $I_j(t) \equiv \int^1_{-1} \exp(ts)s^{j}\text{d}s$. For $j \in \mathbb{N}$ and $t \ne 0$, one obtains 
\begin{align*}
    I_{j}(t) = \frac{1}{t} \left(e^t + (-1)^{\mathds{1}\{j \mod 2 = 0\}}e^{-t} \right) - \frac{j}{t}I_{j-1}(t),
\end{align*}
with $I_0 (t) = \frac{1}{t}(e^t - e^{-t})$. For $t = 0$, $I_j(0) = \frac{2}{j+1} \mathds{1} \{ j \mod 2 = 0\}$ for all $j \in \mathbb{N} \cup \{0\}$. 

\selectlanguage{english}

\nocite{*}

\addcontentsline{toc}{section}{6~~~~References.}

\bibliography{literature.bib} 
\bibliographystyle{ecta} 



\newpage
\appendix
\addcontentsline{toc}{section}{Online Appendices} 
\part{Online Appendices}
\parttoc
\section{Identification}
\subsection{Proof of Lemma 1}\label{appB1}
    \begin{proof}
        Necessity is obvious. For sufficiency, recall that the elements of the RKHS of a continuous kernel $K: H^2 \to \mathbb{R}$ are continuous functions over $H^2$ (see \citet{Wendland_2004}). A simple continuity argument establishes that for any $f \in C(\text{Int}(H))$, $f = 0 ~ \lambda-$a.e. is equivalent to $f(s) = 0, ~ s \in H$. Setting $f(s) = \E_{\P_1}[K(s,U)] - \E_{\P_2}[K(s,U)]$ and applying Definition 1 yields the result. 
    \end{proof}
\subsection{Proof of Lemma 2}\label{proof_lemma2}
\begin{proof}
Necessity follows, because if $P = P_{I_1} \times P_{I_2}$, we have
\begin{align*}
    \int K(s, (U_{I_1},U_{I_2})) \text{d}P(U_{I_1}, U_{I_2}) & = \int \int K(s, (U_{I_1}, U_{I_2}))\text{d}P_{I_2}(U_{I_2}) \text{d}P_{I_1}(U_{I_1})\\ &= \E_P\left[\int K(s, (U_{I_1}, U_{I_2}))\text{d}P_{I_2}(U_{I_2})\right].
\end{align*} For sufficiency, recall the definition of the map $\Phi(P) =E_P[K(s,U)]$. For any $P \in \mathcal{M}_1(\mathcal{U})$, suppose the following equation holds:
\begin{align*}
    \Phi(P) = \E_P\left[\int K(\cdot,(U_{I_1},U_{I_2}))\text{d}P_{I_2}(U_{I_2})\right].
\end{align*}
Construct a new measure $\tilde{P} = P_{I_1} \times P_{I_2} \in \mathcal{M}_1(\mathcal{U})$, and observe that, under this measure,
\begin{align*}
    \Phi(\tilde{P}) &= \E_{\tilde{P}}\left[ \int K(\cdot,(U_{I_1},U_{I_2}))\text{d}\tilde{P}_{I_2}(U_{I_2})\right] \\
    &= \E_{P_{I_1}}\left[ \int K(\cdot,(U_{I_1},U_{I_2}))\text{d}{P}_{I_2}(U_{I_2})\right]  \\
    &= \Phi(P),
\end{align*}
from where the claim of the Lemma follows by injectivity of $\Phi(P)$ and the continuity argument in the Proof of Lemma 1.
\end{proof}
\subsection{Proof of Lemma 3}

    

The argument is the same, as in the Proof of Lemma 2, where one observes that under any $P\in \mathcal{M}_1(\mathcal{U}),$ by law of iterated expectation, the joint MGF of $(U_{I_{j}})^3_{j=1}$ rewrites as $$
M_{P_{I_1 \cup I_2\cup I_3}}(s) = \E_{P}[\exp(\sum^3_{j=2} s'_{I_j}U_{I_j}) \E_{P} [\exp(s'_{I_1}U_{I_1})|U_{I_2 \cup I_3}]],$$
and if $U_{I_1} \indep_P U_{I_2}|U_{I_3},$ then \begin{align*}\E_{P}[\exp(\sum^3_{j=2} s'_{I_j}U_{I_j}) \E_{P} [\exp(s'_{I_1}U_{I_1})|U_{I_2 \cup I_3}]] = \\  \E_{P}[\exp(\sum^3_{j=2} s'_{I_j}U_{I_j}) \E_{P} [\exp(s'_{I_1}U_{I_1})|U_{I_3}]],\end{align*}
where the RHS is the MGF of the measure constructed as $\tilde{P}_{I_1|I_3}\times P_{I_2,I_3}$, under which $U_{I_1} \indep U_{I_2}|U_{I_3}$, where RCDs are well-defined because $U$ is Borel, see \citet{kallenberg2002foundations}.
\section{Estimation}
\subsection{Proof of Theorem \ref{theorem_dual}}
We first construct the dual problem to (P) in the sense of \citet{shapiro}.
\begin{prop}\label{prop1}
    Under Assumption 1.i, the dual problem to $(P)$ in the sense of \citet{shapiro} is given by
\begin{align}\label{dual_glp}
    (D): \quad \beta^*(\theta) \equiv \sup_{\mu \in \mathcal{M}^+(T)} \int b(t) \text{d} \mu(t), \quad \text{s.t.: ~ ~ } c =  \int a(t) \text{d} \mu(t)~  \text{ $\nu-$a.s.}
\end{align}
In particular,
\begin{align*}
        \beta^*(\theta) = \sup_{\mu \in \mathcal{M}^+(T)} \inf_{x \in X} \langle c, x \rangle + \int b(t) - \langle a(t), x \rangle \text{d}\mu(t) 
\end{align*}
\end{prop}
\begin{proof}
Define $Y \equiv C(\mathcal{T})$, and consider the cone of non-positive functions $K \equiv \{f \in Y: f(t) \leq 0, \forall t \in \mathcal{T}\}$. By Riesz Representation Theorem, since $T$ is a compact metric space, the dual space $Y^*$ is naturally identified with the space of finite Borel measures over $T$, $\mathcal{M}(T)$. Moreover, it is clear that the polar cone $K_{-} \equiv \{y^* \in Y^*: \langle y^*, y \rangle \leq 0 ~ \forall y \in Y \} \subseteq Y^*$ is naturally identified with the set of non-negative finite Borel measures over $T$, $\mathcal{M}^+(T) \subseteq \mathcal{M}(T)$. 

First, let us show that because $a: \mathcal{T} \to X$ is continuous, $\sup_{t \in \mathcal{T}} ||a(t)||_2 < \infty$. This follows by triangle inequality. Consider $t, t_1 \in \mathcal{T}$, and observe that 
\begin{align*}
    \left|||a(t_1)||_2 - ||a(t)||_2\right| \leq ||a(t_1) - a(t)||_2 \to 0,
\end{align*}
as $t_1 \to t$. So, $t \to ||a(t)||_2$ is also continuous over $\mathcal{T}$, and thus $\sup_{t \in \mathcal{T}} ||a(t)||_2 < \infty$ by Weirstrass theorem. 

The restrictions of the primary problem can be rewritten as $b - Ax \in K$ for a linear operator $A: X \to Y$, which acts as $x \to \langle x, a(\cdot) \rangle$. It is straightforward to observe that $A$ is continuous when $a$ is continuous, as $||A (x + h) - Ax ||_\infty = ||A h||_\infty = \sup_{t \in \mathcal{T}}|\langle h, a(t) \rangle| \leq \sup_{t \in \mathcal{T}}  ||h||_2 ||a(t)||_2 \leq ||h||_2 \sup_{t \in \mathcal{T}} ||a(t)||_2 \to 0$ as $||h||_2 \to 0$, because $\sup_{t \in \mathcal{T}} ||a(t)||_2 < \infty$.

Let us characterize the adjoint of $A$, namely the operator $A^*: Y^* \to X^*$, such that $\langle y^*, Ax \rangle = \langle A^*y^*, x\rangle$ for all $x \in X, y^* \in Y^*$. The LHS can be written as $\int \langle a(t), x\rangle \text{d}\mu(t)$, where $\mu$ represents $y^*$. We wish to show that the operator $A^*$ maps $\mu \to \int a(t) \text{d}\mu(t)$. 

We consider the latter integral in Bochner sense. For that, observe that $a(t)$ is strongly measurable by Pettis measurability theorem, because i) it is weakly measurable: $t\to\langle a(t), x^*\rangle$ is continuous for any $x^* \in X^*$; and ii) $X$ is separable. Moreover, by the fact that $\int ||a(t)|| \text{d}|\mu|(t) < \infty$, $a$ is then also Bochner-integrable (see Chapter 1 in \citet{hytonen2016analysis}). Thus, $\int a(t) \text{d}\mu(t) \in X$ is well-defined\footnote{if $\mu$ is a signed measure, define the integral as $\int f \text{d}\mu_{+} - \int f \text{d}\mu_{-}$}. To prove that it is indeed the sought adjoint, we need the following lemma. 
    \begin{lemma}
        Consider a measurable space $(\Omega, \mathcal{F})$ and consider two signed measures $\mu, \nu$ over it. Moreover, suppose that $V$ is a Banach space. If a strongly measurable $f: \Omega \to V$ is $|\mu| + |\nu|$-Bochner-integrable, then $\int f \text{d}\mu + \int f \text{d}\mu = \int f \text{d} (\mu + \nu)$. 
    \end{lemma}
    \begin{proof}
        Construct a sequence of simple functions $\{f_n\}$, such that $\int ||f_n - f|| \text{d} (|\mu| + |\nu|) \to 0$. By the properties of Lebesgue integral and positivity of $|\nu|, |\mu|$, it then follows that both $\int ||f_n - f|| \text{d} |\mu| \to 0$, and $\int ||f_n -f|| \text{d}|\nu| \to 0$. This implies $\int ||f_n - f|| \text{d} \mu_{\pm} \to 0$ and $\int ||f_n - f|| \text{d} \nu_{\pm} \to 0$. Per usual arguments involved in the constructions of a Bochner integral (see \citet{hytonen2016analysis}), each of the sequences $\int f_n  \text{d} \mu_{\pm}, \int f_n  \text{d} \nu_{\pm}, \int f_n  \text{d} (\nu + \mu)_{\pm}$ then converges to a respective element in $V$. Because linearity holds for simple functions, we have $\int f_n  \text{d} \mu_{\pm} + \int f_n  \text{d} \nu_{\pm}  = \int f_n  \text{d} (\nu + \mu)_\pm$ for each $n \in \mathbb{N}$. The continuity of $g(x,y,z) = x + y - z$ in $V^3$ yields the result. 
    \end{proof}
    The above Lemma establishes additivity of $A^*$. Proving that $A^*$ commutes with scalar multiplication is trivial.  Boundedness is obtained via $||A^*\mu|| = ||\int a(t) \text{d} \mu(t)|| \leq \int ||a(t)|| \text{d} |\mu|(t) \leq \sup_{t \in \mathcal{T}} ||a(t)|| \cdot ||\mu||_{TV}$. Thus, $||A^*||_{op} \leq \sup_{t \in \mathcal{T}} ||a(t)|| < 0$. So, $A^* \in \mathcal{L}(Y^*, X^*)$. Finally, the discussion on page 15 in \citet{hytonen2016analysis} allows to conclude that, indeed,
    \begin{align*}
        \langle A^* y^*, x\rangle = \langle \int a(t)\text{d}\mu(t), x\rangle = \int \langle a(t), x\rangle \text{d}\mu(t) = \langle y^*, Ax\rangle.
    \end{align*}
    Finally, it follows from equation 2.6 in \citet{shapiro}\footnote{This equation contains a typo: $\min$ should be $\max$, see the rest of the chapter and the enclosing book.}, that the dual to 
    \begin{align*}
        \inf_{x \in X} \langle c,x \rangle, \quad \text{s.t.: } b - Ax \in K,
    \end{align*}
    is
    \begin{align*}
        \sup_{y^* \in K_{-}} \langle y^*, b\rangle, \quad \text{s.t.: } A^* y^* = c.
    \end{align*}
    Plugging in the definitions yields the claim of the proposition. 
\end{proof}
We now proceed to prove the main result. 
\begin{proof}
    We first show that Slater Condition (see \citet{shapiro}) holds in (P) problem under Assumption 1. Showing this amounts to proving that $\exists x \in X$, such that
    \begin{align*}
        \langle a(t),x \rangle - b(t) > 0 ~ \forall t \in \mathcal{T}.
    \end{align*}
    Observe that, for any $x \in X$, $\langle a(t),x \rangle = x(s^*) + \int_{S \setminus \{s^*\}} a(t)(s) x(s)\text{d}\nu(s)$. Consider $x = \mathds{1}_{\{s^*\}} x^*$. As $s^* \in \mathcal{S},$ such $x$ is measurable. Moreover, since $\nu(\cdot)$ is finite by assumption, $\nu(\{s^*\}) < \infty$, and so $x \in X$. We get 
    \begin{align*}
        \langle a(t),x \rangle - b(t) = x^* \nu(\{s^*\}) - b(t).
    \end{align*}
    By compactness of $\mathcal{T}$ and 1.i, $\sup_{t \in \mathcal{T}} b(t) < \infty$, so $x^* = \frac{\sup_{t \in \mathcal{T}} b(t)}{\nu(\{s^*\})} + 1$ yields the required point. Applying Theorem 2.4. in \citet{shapiro} establishes strong duality $\beta(\theta) = \beta^*(\theta)$, and the equivalence $|\beta^*(\theta)| < \infty \iff \mathcal{A}^*(\theta) \ne \emptyset$. 
    
    The rest is established as follows. Let $\mathcal{L}_\infty(x;\theta) \equiv \langle c, x\rangle + \left(\max_{t \in \mathcal{T}} b(t) - \langle a(t), x \rangle \right)^+$ for any $x \in X$, and $\theta$ satisfying Assumption 1. Denote the feasible set of $(P)$ by $\Theta_I(\theta)$. We have $\mathcal{L}_{\infty}(x, \theta) = \langle c,x \rangle$ for any $x \in \Theta_I(\theta)$. Because $\Theta_I(\theta) \subseteq X$, we have
    \begin{align}\label{one_direction_}
        \beta_\infty(\theta) = \inf_{x \in X} \mathcal{L}_\infty(x;\theta) \leq \inf_{x \in \Theta_I(\theta)} \mathcal{L}_{\infty}(x;\theta) = \inf_{x \in \Theta_I(\theta)} \langle c,x \rangle = \beta(\theta).
    \end{align}
    Now suppose that $\mathcal{A}^*(\theta)$ is non-empty, which is equivalent to Assumption 1.iii by the above result. By definition of the dual program (see e.g. \citet{shapiro}), the dual problem (D) is equivalent to 
    \begin{align*}
        \sup_{\mu \in \mathcal{M}^+(T)} \inf_{x \in X} \langle c, x \rangle + \int b(t) - \langle a(t), x \rangle \text{d}\mu(t)  
    \end{align*}
    So, for any $\mu^* \in \mathcal{A}^*(\theta)$ we have
    \begin{align}
        \beta^*(\theta) = \inf_{x \in X} \langle c, x \rangle + \int b(t) - \langle a(t), x \rangle \text{d}\mu^*(t),
    \end{align}
    implying that
    \begin{align}\label{other_direction_}
        \beta^*(\theta) \leq \langle c, x \rangle + \int b(t) - \langle a(t), x \rangle \text{d}\mu^*(t) ~ \forall x \in X
    \end{align}
    Recalling that $\mu^* \in \mathcal{M}^+(T)$ and using properties of Lebesgue integral, we obtain
    \begin{align}\label{ineqs_}
        \int b(t) - \langle a(t), x \rangle \text{d}\mu^*(t) &\leq  \int (b(t) - \langle a(t), x \rangle)^+ \text{d}\mu^*(t) \leq \\ \notag \int \left[\sup_{t \in \mathcal{T}} (b(t) - \langle a(t), x \rangle)^+ \right]\text{d}\mu^*(t) 
        &\leq \sup_{t \in \mathcal{T}} (b(t) - \langle a(t), x \rangle)^+ \int \text{d}\mu^*(t),
    \end{align}
    for any $x \in X$. Using Assumption 1.ii, we observe that, for any $\mu^*$ satisfying the constraints of (D), $\int \text{d}\mu^* = 1$. Adding $\langle c, x\rangle$ to both sides of \eqref{ineqs_}, taking infinum over $X$ and combining with \eqref{other_direction_}, yields 
    \begin{align*}
        \beta^*(\theta) \leq \beta_\infty(\theta),
    \end{align*}
    which, combined with \eqref{one_direction_}, concludes the proof of the Lemma.
\end{proof}
\subsection{Proof of Lemma \ref{lemma_weak_compactif}}
Fix some $\theta$ satisfying Assumption 1. Because $\mathcal{L}_\infty(x;\theta)$ is a continuous convex functional of $x \in X$, and $X$ is a Hilbert space (thus reflexive), part i) follows from Proposition 1.2 on p.35 in \citet{ekeland1999convex}, which is a consequence of Banach-Alaoglu Theorem. Part ii) follows directly from the definition of $\Delta(\cdot)$. For Part iii), consider the problem in \eqref{def_unc}. By Theorem \ref{theorem_dual}, there exists $\{x_n\} \subset X$, such that $\mathcal{L}_\infty(x_n;\theta) \downarrow \beta^*(\theta)$. Setting $\gamma_n = ||x_n||,$ and using Part ii) yields the result. 

\subsection{Proof of Theorem \ref{theor_general}}
\begin{proof}
    By definition, 
    \begin{align*}
        \beta_{J}(\hat{\theta};\gamma) - \beta(\theta) & = \underbrace{\beta_{J}(\hat{\theta}; \gamma) - \beta_{J}(\theta; \gamma)}_{I_1} \\
        &+ \underbrace{\beta_{J}(\theta; \gamma) - \beta_{\infty}(\theta;\gamma)}_{I_2} \\
        &+ \Delta(\theta; \gamma),
    \end{align*}
    where the last two terms are positive, so
    \begin{align}\label{eq_withterms}
\beta_{J}(\hat{\theta};\gamma) - \beta(\theta) \in [-|I_1|; |I_1| + I_2 + \Delta(\theta;\gamma)]
    \end{align}
    For the first term, by Cauchy-Schwarz and the triangle inequality,
    \begin{align}\notag
        |I_1| &= \left|\inf_{x \in B_{\gamma} \cap \mathbf{P}_{J}} \mathcal{L}_\infty(x;\hat{\theta}) - \inf_{x \in B_{\gamma} \cap \mathbf{P}_{J}} \mathcal{L}_\infty(x;\theta)\right| \\\notag
         &\leq \sup_{x \in B_{\gamma}}\left|\mathcal{L}_\infty(x;\hat{\theta}) - \mathcal{L}_\infty(x;\theta)\right|  \\\notag
         &\leq \sup_{x \in B_{\gamma}}\left| \langle \hat{c} - c, x\rangle\right| + \sup_{t \in \mathcal{T}} \left|\hat{b}(t) - b(t) + \langle a(t) - \hat{a}(t),x\rangle\right| \\\notag
         &\leq \sup_{x \in B_{\gamma}} ||x||\left(\sup_{t \in \mathcal{T}} ||\hat{a}(t) - a(t)||+ ||\hat{c} - c||\right) + \sup_{t \in \mathcal{T}}||\hat{b}(t) -b(t)|| 
         \\ \label{I1term}
         &= \gamma(\delta^a + \delta^c) + \delta^b.
    \end{align}
    Combining \eqref{eq_withterms} with \eqref{I1term} and using \eqref{polynoms_approx_1} for $I_2$ yields the first claim of the theorem. 

    For the second part, let us state two auxiliary results.
\begin{lemma}\label{lemma_aux1}
    Suppose $d, K, s \in \mathbb{N}$, with $K \geq s$, and $f \in \mathcal{L}^2([-1,1]^d)$. There exists a constant $M > 0$, independent of $K$ and $f$, such that 
    \begin{align*}
        ||(I - \mathbf{p}_K)f||_2 \leq M \omega^{s}(f,K^{-1})_2,
    \end{align*}
    where $\omega^{r}(f, K^{-1})_2$ is the $\mathcal{L}^2-$integrated modulus of smoothness of $f$ of degree $r$, with the step size $K^{-1}$.
\end{lemma}
\begin{proof}
    The claim follows from Theorems 4.1.1. and 12.1.1 in \citet{ditzian2012moduli}. 
\end{proof}
\begin{lemma}\label{lemma_aux2}
    For some $C > 0$ independent of  $f \in W^{s,2}((-1,1)^d)$, we have $\omega^{s}(f,t)_2 \leq t^{s} C ||f||_{W^{s,2}}$ for all $t > 0$. 
\end{lemma}
\begin{proof}
    Because the cube is a Lipschitz domain, we can use the Sobolev extension theorem to extend $f$ to $\tilde{f}$ on $\mathbb{R}^d$, with $||\tilde{f}||_{W^s(\mathbb{R}^d)} \leq C ||f||_{W^s((-1,1)^d)}$ for an absolute constant $C > 0$. Since $\tilde{f}$ coincides with $f$ on $(-1,1)^d,$ it follows that $\omega^s(f, t)_2 \leq \omega^s(\tilde{f}, t)_2$, where the latter is upper bounded by $\tilde{C} t^s||\tilde{f}||_{W^{s,2}(\mathbb{R}^d)}$ for an absolute constant $\tilde{C} > 0$ by Property 4 in \citet{KOLOMOITSEV2020105423} for any $t > 0$. 
\end{proof}
The bound on $\delta^p$ follows by combining Lemmata \ref{lemma_aux1},\ref{lemma_aux2} with \eqref{polynoms_approx_1} and Assumption 3. 
\end{proof}
\subsection{Proof of Theorem \ref{theor_main_cons}}
\begin{proof}
    To prove the Theorem, we need to show that, firstly, the construction in Assumption 4 satisfies, uniformly over $\P,$ $\delta_n^{k} = O_\P(1/{\sqrt{n}})$ for $k = a,c$. 
    
    We first consider $a$. Observe that under the constructions in Section 2.3 $a$ is fixed, so the corresponding increments are $0$. For Section 2.4, it suffices to study one independence restriction. Suppose $T = (T_1, T_2),$ where one assumes $T_1 \indep T_2$, and $W_i$ for $i \in [n]$ are copies of $T_2$. Denote the supports of $T_j$ by $\mathcal{T}_j$. Plugging in the construction, with $S = H_d(0)$, and $\nu = \lambda$, 
    \begin{align*}
        \sqrt{n}\delta^a_n = \sup_{t \in \mathcal{T}} ||\sqrt{n}(\hat{a}(t) - a(t))|| = \sup_{t_1 \in \mathcal{T}_1}\left(\int_{H_d(0)} (\sqrt{n}(\P_{2n} - \P_2)K(s,t_1, \cdot))^2\text{d}s\right)^{1/2}
    \end{align*}
    For a Matérn Kernel $K: \mathbb{R}^{2d} \to \mathbb{R}$, it is well known that the map $t \to K(s,t)$ is in $H^{d/2 + \nu}(\mathbb{R}^d)$ for any $s \in \mathbb{R}^d$ (see e.g. \citet{multires_matrix}). By definition of Matérn kernel, the restriction of this map to the map $t_2 \to K(s,(t_1, t_2))$ rewrites as 
    \begin{align*}
        K((s_1,s_2),(t_1, t_2)) = C_\nu (||s_1 - t_1||^2 + ||s_2 - t_2||^2)^{\nu/2}K_\nu\left(\sqrt{||s_1 - t_1||^2 + ||s_2 - t_2||^2 }\right). 
    \end{align*}
    If $s_1 = t_1$, this clearly degenerates to the Matérn kernel over $R^{2d_2}$, so the map is in $H^{d_2/2 + \nu}(\mathbb{R}^{d_2})$. Otherwise, $||s_1 - t_1||^2 > 0$, and so the above map has an even higher smoothness, because $K_\nu(r)$ is analytic apart from $r = 0$. Thus, for any $s, t_1$ the map $t_2 \to K(s,(t_1, t_2))$ is in $H^{d_2/2 + \nu}(\mathbb{R}^{d_2})$. By continuity of $K(s,t),$ the non-empty set
    \begin{align*}
        \mathcal{F} \equiv \{t_2 \to \xi(t_2)K(s,(t_1, t_2)): s \in H_d(0), t_1 \in \mathcal{T}_1\} \subseteq H^{d_2/2 + \nu}(\mathbb{R}^{d_2})
    \end{align*}
    is bounded for an appropriate analytic mollifier $\xi(t_2) \in [0,1]$ that vanishes on $\mathcal{T}_2^{\varepsilon}\setminus \mathcal{T}_2$, and becomes $\xi(t_2) = 1$ on $t_2 \in \mathcal{T}_2$, for some $\varepsilon > 0$. Define $\mathcal{F}|{\mathcal{T}_2}$ to be the set of restrictions of functions from $\mathcal{F}$ to the support of $T_2$. Corollary 2 in \citet{nickl2007bracketing} establishes that over $\mathcal{P}$, we have, for some constant $C_2 > 0$,
\begin{align}
    \sup_{\P_2 \in \mathcal{P}_2} \log N_{[]}(\varepsilon, \mathcal{F}|{\mathcal{T}_2})  \leq C_2 \varepsilon^{-d_2/(d_2/2 + \nu)},  
\end{align}
    which satisfies the entropy bound condition in Theorem 2.8.4 in \citet{van1996weak} for any $\nu > 0$. Thus, the class $\mathcal{F}|\mathcal{T}_2$ is uniformly Donsker and uniformly pre-Gaussian, and so
    \begin{align*}
        \mathbb{G}_{n}\equiv\sqrt{n}(\P_{2n} - \P_2)
    \end{align*}
    converges weakly in $\ell^\infty(\mathcal{F}|\mathcal{T}_2)$ to a tight, Borel measurable version of the Brownian bridge $\mathbb G_{\P}$, uniformly in $\P \in \mathcal{P}$. The conclusion then follows by e.g. bounding $\sqrt{n}\delta^a_n$ asymptotically with a supremum over a Brownian bridge. The argument for $\delta^a_n$ is identical if $K$ is the exponential kernel, observing that it gives rise to analytic functions. The result for $\hat{c}$ is established similarly. The results in the Theorem then follow by combining the rates for $\delta_n^{\bullet}$ with the bounds in Theorem \ref{theor_general}.
\end{proof}
\subsection{Proof of Lemma \ref{lemma_highdim}}
\begin{proof}
    Part i) of the Lemma follows by observing that $\mathcal{L}_\infty(x;\theta_n)$ is a continuous convex functional over $X_n$, and $B_\gamma$ is a closed convex subset of $X_n$, so Proposition 1.2 on p.35 in \citet{ekeland1999convex} yuelds the result.  

    To prove part ii), observe that $\Delta(\theta_n, \gamma) = f(n, \overline{\gamma}, \tilde{\gamma})$ for some fixed $f: \mathbb{R}^3 \to \mathbb{R}$. We need to show that $f$ is non-increasing in each of its arguments. First, fix some $\gamma$. For any $n \in \mathbb{N}$, by Assumption 5.ii, $\beta(\theta_{n-1}) =\beta(\theta_n)$, whereas $\beta(\theta_n;\gamma) \leq \beta(\theta_{n-1}; \gamma)$ by Assumption 5.iii. For a fixed $n$ and $\gamma_1 \geq \gamma_0$, we have $\beta(\theta_n;\gamma_1) \leq \beta(\theta_n;\gamma_0)$ by definition of $\beta(\cdot;\cdot)$.
    
    Part iii) follows from Parts ii) and i). This concludes the proof of the lemma.
\end{proof}
\subsection{Proof of Theorem \ref{theor_main_cons_highdim}}
\begin{proof}
    It is straightforward to observe that Theorem \ref{theor_main_cons} yields control over $\delta^a_n$ and the fixed continuous part of $\delta^c_n$ in Teorem \ref{theor_general_highdim}. It remains to establish control control over the rate of the high-dimensional part of $\delta^c_n$. This is done via the following well-known consequence of Hoeffding inequality combined with the union bound. 
    \begin{lemma}
        Suppose $\{U_i\}^n_{i=1}\subseteq \mathbb{R}^p$ are i.i.d., with $U_i \in \mathbb{R}^p$, and $||U_i - \E[U_i]|| < u$ a.s. for some $u > 0$ and any $i \in [n]$. Then, for any $\varepsilon \in (0;1]$,
        \begin{align*}
            \P\left[||\frac{1}{n}\sum^n_{i=1}U_i - \E[U_i]||_\infty \geq u \sqrt{2}\sqrt{\frac{\log p + \log 2/\varepsilon}{n}}\right] \leq \varepsilon
        \end{align*}
    \end{lemma}
    In our context, we get, for any $n \in \mathbb{N}$, any $\P \in \mathcal{P}$ and any $\varepsilon \in (0;1]$, 
    \begin{align*}
        \P\left[\lVert\hat{\overline{c}}_n - \overline{c}_n\rVert_\infty \geq \overline{C}\sqrt{2}\sqrt{{\frac{\log l_n + \log 2/\varepsilon}{n}}}\right] \leq \varepsilon,
    \end{align*}
    which yields the result.  
\end{proof}
\section{Bias characterization}
\subsection{Proof of Proposition \ref{theor_charact}}
\begin{proof}
We first prove the first claim of the Proposition. For ease of notation, we normalize throughout $\lambda (\mathcal{T}) = 1$. It will suffice to concentrate on a single-coordinate $\gamma$, as an extension to the procedure in Section \ref{section_highdim} is straightforward.

First, observe that we can treat $A$ as an operator $A: X \to \mathcal{L}^2(\lambda)$. By Assumption 1.i, and because an equivalence class in any $\mathcal{L}^2$ space has at most one continuous representative, we have $\{x \in X: (Ax)(T) - b(T) \geq 0, ~ \forall T \in \mathcal{T}\} = \{x \in X: (Ax)(T) - b(T) \geq 0, ~ \lambda-\text{a.e.}\}$. We can also observe that, by condition iv), $\mathbb{E}_\P\big [(Ax)(T)\big ] = \E_\P\big [[(Ax)(T)]_{\sim}\big]$, where the latter is the expectation of some representative of the equivalence class in $\mathcal{L}^2(\lambda)$. Therefore, we can write 
    \begin{align*}
        (P) \quad &\inf_{x \in X} \mathbb{E}_\P[ (Ax)(T)] ~ \text{ s.t.: } (Ax)(T) \geq b(T), ~ \forall T \in \mathcal{T} = \\
        (P^*) \quad &\inf_{x \in X} \mathbb{E}_\P\big [ [(Ax)(T)]_{\sim}\big] ~ \text{ s.t.: } Ax  \geq b,~  \lambda - \text{a.e.}
    \end{align*}
    We now observe that the value of $(P^*)$ is greater or equal than the value of the following relaxation:
    \begin{align*}
        (P^{**}) \quad \beta^{**}(\theta) \equiv \inf_{f \in \mathcal{L}^2(\lambda)} ~ \mathbb{E}_\P[f(T)], \quad \text{s.t.: } f \in \text{Cl}\left(R(A) \cap \{f \in \mathcal{L}^2(\lambda): ~f \geq b ~ \text{a.e.}\}\right)
    \end{align*}
    To see that the converse also holds, observe that, by definition of infimum, for any $\eta > 0$ there exists a feasible $f$ in $(P^{**})$, such that $\E_\P[f(T)] < \beta^{**}(\theta) + \eta$. Any such $f$ is the $\mathcal{L}^2$ limit of a sequence $Ax_n$ that is feasible in $(P^*)$, and by Cauchy-Schwarz and using assumption C4,
    \begin{align*}
        |\E_\P [f(T)] - \E_\P [(Ax_n)(T)]| \leq  \left \| \frac{\partial \P}{\partial \lambda}\right \|_2  \| f - Ax_n\|_2,
    \end{align*}
    so for any $\eta > 0$ there exists $x \in X$, such that $Ax$ is feasible in $(P^*)$, and $\E_\P[(Ax)(T)] < \beta^{**}(\theta) + 2\eta$. Thus, $\beta(\theta) = \beta^{**}(\theta)$. This concludes the proof of the first part of the proposition.

    For the second claim, fix some $f^*$ as in the statement of the theorem. Because it is feasible in problem \eqref{eq_solvp}, and as the set $\{f \in \mathcal{L}^2(\lambda): f \geq b~ \lambda-\text{a.e.}\}$ is closed in $\mathcal{L}^2$ and, moreover, the closure of the intersection is contained in the intersection of the closures, it must be that $f^* \geq b ~~  \lambda-\text{a.e.}$. By properties of essential supremum, for any $\varepsilon > 0$, one can find $x^* \in B_\gamma$, such that $|\textbf{A}x^* - f^*| < \inf_{x \in B_\gamma} \text{ess} \sup |\textbf{A}x - f^*| + \varepsilon/2$ $\lambda-$a.e. Such $x^*$ then satisfies $\textbf{A}x^* \geq f^* - \inf_{x \in B_\gamma} \text{ess} \sup |\textbf{A}x - f^*| - \varepsilon$. Therefore, because $\textbf{A}x^*$ is a continuous function, $\textbf{A}x^* +  \inf_{x \in B_\gamma} \text{ess} \sup |\textbf{A}x - f^*| + \varepsilon \geq b(T)$. Recall that $S$ contains an atom $s^*$ with $a(\cdot)(s^*) = 1$ under Asssumption 1.ii. It follows that we can increase take $\tilde{x}(s) = x^*(s) + (\inf_{x \in B_\gamma} \text{ess} \sup |\textbf{A}x - f^*| + \varepsilon )\mathds{1}\{s = s^*\},$ which then gives $\textbf{A}\tilde{x} = \textbf{A}x^* + (\inf_{x \in B_\gamma} \text{ess} \sup |\textbf{A}x - f^*| + \varepsilon ) \geq b$, has $||\tilde{x}|| \leq \gamma + (\inf_{x \in B_\gamma} \text{ess} \sup |\textbf{A}x - f^*| + \varepsilon) $ and so is feasible in (P).   
\end{proof}
\subsection{Proof of Lemma \ref{lemma_splines}}
\begin{proof}
We proceed in 5 steps. 

    \paragraph*{Step 1} \citet{ober2023estimating} shows that we may restrict ourselves to dual potentials $f_i^*$ being $L-$Lipschitz, and having the norm $||f^*_i||_\infty \leq 2||b||_\infty$ for $i = 1,2$.
    \paragraph*{Step 2} Using McShane extension, we obtain continuous maps $f_i^*: H_{d_i}(0) \to \mathbb{R}$ that agree with the potentials from Step 1 over $\mathcal{T}_i$, are $L-$Lipschitz over $H_{d_i}(0)$, and have the norm $\sup_{t \in H_{d_i}(0)} |f^*_i(t)| \leq 2||b||_\infty + L \equiv M$. 
    \paragraph*{Step 3} We want to show that 
    \begin{align}
        \min_{\alpha \in \mathbb{R}^{k^{d_i}}: ||\alpha||_1 \leq k^{d_i} M} \sup_{t \in H_{d_i}(0)} |f^*_{i}(t) - \alpha'\mathbf{B}^{(k_i)}_i(t)| \leq \underbrace{\min \left\{\frac{2L\sqrt{d_i}}{k_i - 1},M\right\}}_{\varrho_i}
    \end{align}
     The knots of linear splines partition $H_{d_i}(0)$ into $(k_i-1)^{d_i}$ cubes. We consider the interpolation procedure, under which the interpolant $\hat{f}^*_i(t)$ is equal to $f^*_i(t)$ function on the boundaries of such cubes. For any $t \in H_{d_i}(0)$, take one of the resulting cubes, $Q$, to which it belongs, with vertices $\{t_v\}_{v \in \{0,1\}^{d_i}}$. The multilinear interpolant we get has the form $\hat{f}^*_i(t) = \sum_{v \in \{0,1\}^{d_i}} f^*_i(t_v) B_\nu(t),$ where $B_\nu(t) \in [0,1]$ are products of univariate hat-functions, and $\sum_{\nu \in \{0,1\}^{d_i}} B_\nu(x) = 1$. By triangle inequality, $$|\hat{f}^*_i(t) - f^*_i(t)| \leq \sum_{v \in \{0,1\}^{d}}B_\nu(t)|\hat{f}^*_i(t_v) - f^*_i(t)|,$$
and by Lipschitz continuity and Hölder inequality, we get  $f^*_i$, $$\sum_{v \in \{0,1\}^{d}}B_\nu(t)|\hat{f}^*_i(t_v) - f^*_i(t)| \leq L\sum_{v \in \{0,1\}^{d}}B_\nu(t)||t_v - t|| \leq L h \sqrt{d},$$
where $h = \frac{2}{k_i - 1}$ is the side of the cube. The result follows from observing that in our construction $||\alpha||_1 \leq k^{d_i}||\alpha||_\infty \leq k^{d_i}M$. The minimum with $M$ is from $\alpha = 0$.
   \paragraph*{Step 4} Let $x_0 \in X$ be such that $x_0(s) = 0$ for $s \notin (H \setminus \{s^*\})$ and otherwise equal to the coefficients of the linear interpolation obtained in Step 3 or to $0$ if the corresponding bound $\varrho_i$ is equal to $M$. From Step 1-3 and feasibility of dual potentials it follows that $x_1 = x_0 + (\varrho_1 + \varrho_2)\mathds{1}\{s = s^*\}$ satisfies:
   \begin{align*}
       (\mathbf{A}_n^{(ot)} x_1)(T) \geq b(T) ~ \forall T \in \mathcal{T},
   \end{align*}
   and $||\tilde{x}_1|| = 0,$ and by triangle inequality and conclusion of Step 1, $||\overline{x}_1||_1 \leq ||\overline{x}_0||_1 + \sum^2_{i=1}\varrho_i \leq M(2+ \sum^2_{i=1}k_i^{d_i})$ i.e. $x_1$ is feasible in $(P)$. Moreover, 
   \begin{align*}
       \E_\P[(\mathbf{A}_n^{(ot)} x_1)(T) - f^*_1(T_1) - f_2^*(T_2)] \leq 2 (\varrho_1 + \varrho_2) \leq 2\sum^{2}_{i=1}\frac{L\sqrt{d_i}}{k_i-1}.
   \end{align*}
   \paragraph*{Step 5} The proof is completed by observing that because the kernels $K_i$ used in the construction of $\theta^{(ot)}_n$ are characteristic, it follows from Lemma \ref{lemma_identif}, OT duality, and the results in Section 3 that $\beta(\theta_n^{(ot)})$ is equal to the value of the OT for any $n \in \mathbb{N}$.  
\end{proof}
\end{document}